\documentclass[12pt]{amsart} 

 \usepackage{caption}
 \usepackage{subcaption}

\usepackage[dvipsnames,usenames]{xcolor} 

\usepackage{amsfonts,latexsym,wasysym, mathrsfs, mathtools,hhline,color, bm}
\usepackage[all, cmtip]{xy}

\usepackage{url}

\definecolor{hot}{RGB}{65,105,225} 

\usepackage[pagebackref=true,colorlinks=true, linkcolor=hot ,  citecolor=hot, urlcolor=hot]{hyperref}

\usepackage{textcomp}
\usepackage{tipa}

\usepackage{graphicx,enumerate}
\usepackage{enumitem}
\usepackage[normalem]{ulem}
\usepackage[abs]{overpic}

 \usepackage{amsthm}  
 \usepackage{amsmath} 
 \usepackage{amsfonts}  
 \usepackage{amssymb}
 \usepackage{amscd} 
 \usepackage[all]{xy} 

\usepackage{graphicx}
\usepackage{pdfsync}
\usepackage{url}
\usepackage{comment}

\usepackage[linesnumbered,ruled,vlined]{algorithm2e}
\usepackage{cleveref}

\newtheorem{thm}{Theorem}[section]

\theoremstyle{definition}
\newtheorem{example}[thm]{Example}

\newtheorem{lem}[thm]{Lemma}
\newtheorem{prop}[thm]{Proposition}
\newtheorem{cor}[thm]{Corollary}
\theoremstyle{definition}
\newtheorem{defn}[thm]{Definition}
\newtheorem{problem}[thm]{Problem}
\newtheorem{rem}[thm]{Remark}
\numberwithin{equation}{section}

\newcommand{\MM}{\mathcal{M}}
\newcommand{\ff}{p} 
\DeclareMathOperator{\Jac}{Jac}
\DeclareMathOperator{\generic}{gen}

\newcommand{\PP}{\mathbb{P}}

\newcommand{\VV}{V}

\newcommand{\NN}{\mathbb{N}_{\geq 0}}
\newcommand{\CC}{\mathbb{C}}
\newcommand{\AS}{\mathcal{A}}
\newcommand{\ZV}{Z}
\newcommand{\aaa}{t}
\newcommand{\newt}{{\mathrm{Newt}}}
\newcommand{\conv}{{\mathrm{Conv}}}
\newcommand{\mvol}{{\mathrm{MVol}}}
\newcommand{\vol}{{\mathrm{Vol}}}
\newcommand{\init}{{\mathrm{Init}}}
\newcommand{\val}{{\mathrm{val}}}

\title{Estimating Gaussian mixtures using sparse polynomial moment systems}

\author[Lindberg]{Julia Lindberg}
\address{
Julia Lindberg \\
University of Texas-Austin \\ USA
}
\email{julia.lindberg@math.utexas.edu}
\urladdr{\url{https://sites.google.com/view/julialindberg/home}}

\author[Am\'endola]{Carlos Am\'endola}
\address{
Carlos Am\'endola \\
Technische Universit\"at Berlin, Germany
}
\email{amendola@math.tu-berlin.de}
\urladdr{\url{http://www.carlos-amendola.com/}}

\author[Rodriguez]{Jose Israel Rodriguez}
\address{
Jose Israel Rodriguez\\
University of Wisconsin-Madison \\
 USA
} 
\email{jose@math.wisc.edu}
\urladdr{\url{https://sites.google.com/wisc.edu/jose/home}}

\begin{document}

\begin{abstract}
  The method of moments is a classical statistical technique for density estimation that solves a system of moment equations to estimate the parameters of an unknown distribution. A fundamental question critical to understanding identifiability asks how many moment equations are needed to get finitely many solutions and how many solutions there are. We answer this question for classes of Gaussian mixture models using the tools of polyhedral geometry. In addition, we show that a generic Gaussian $k$-mixture model is identifiable from its first $3k+2$ moments. Using these results, we present a homotopy algorithm that performs parameter recovery for high dimensional Gaussian mixture models where the number of paths tracked scales linearly in the dimension. \\
  \hfill \\

  \textbf{Key words: } Gaussian mixtures models, method of moments, parameter recovery, sparse polynomial systems
\end{abstract}

\maketitle

\section{Introduction}

A fundamental problem in statistics is to estimate the parameters of a density from samples. This problem is called \emph{density estimation} and formally it asks: given $n$ samples from an unknown distribution $p$, can we estimate $p$? To have any hope of solving this problem we need to assume our density lives in a family of distributions.  One family of densities known as Gaussian mixture models are a popular choice due to their broad expressive~power.

\begin{thm}\cite[Ch.~3]{Goodfellow-et-al-2016} \label{thm:gmmapproximator}
A Gaussian mixture model is a universal approximator of densities, in the sense that any smooth density can be approximated with any speciﬁc nonzero amount of error by a Gaussian mixture model with enough~components.
\end{thm}

\Cref{thm:gmmapproximator} motivates our study of Gaussian mixture models. These are ubiquitous in the literature with applications in modeling geographic events \cite{GMM_geography}, the spread of COVID-19 \cite{GMM_covid}, the design of planar steel frame structures \cite{GMM_material_science}, speech recognition~\cite{reynolds1992gaussian,zhang1994using, sha2006large}, image segmentation \cite{GMM_image_segmentation} and biometrics \cite{hosseinzadeh2008gaussian}.

A \emph{Gaussian random variable}, $X$, has a probability density function given by
\begin{align*}
    f(x | \mu, \sigma^2) = \frac{1}{\sqrt{2 \pi \sigma^2}} \exp \left( - \frac{(x - \mu)^2}{2 \sigma^2} \right),
\end{align*}
where $\mu \in \mathbb{R}$ is the mean and $\sigma \in \mathbb{R}_{>0}$ is the standard deviation. In this case we write $X \sim \mathcal{N}(\mu, \sigma^2)$. A random variable $X$ is the  \emph{mixture of $k$ Gaussians} if its probability density function is the convex combination of $k$ Gaussian densities. Here we write $X \sim \sum_{\ell=1}^k \lambda_\ell \mathcal{N}(\mu_\ell, \sigma_\ell^2)$ where $\mu_\ell \in \mathbb{R}$, $ \sigma_\ell \in \mathbb{R}_{>0}$ for $\ell \in [k]= \{1,\ldots, k\}$ and $(\lambda_1,\ldots,\lambda_k) \in \Delta_{k-1} = \{ \lambda \in \mathbb{R}_{> 0}^k : \sum_{\ell=1}^k \lambda_\ell = 1 \}$. 
Each $\lambda_\ell$, $\ell \in [k]$, is the mixture weight of the $\ell$th component. 

A \emph{Gaussian $k$-mixture model} is a collection of mixtures of $k$ Gaussian densities. Often one imposes constraints on the means, variances or weights to define such models. 
For example, one might assume that all variances are equal or that the mixture weights are all equal. 
The former is known as a \emph{homoscedastic} mixture model, and the latter is called a \emph{uniform} mixture model.
In this paper we consider three classes of Gaussian mixture models.
\begin{enumerate} 
\item The $\bar \lambda$-weighted model where the mixture weights are fixed for $\bar \lambda \in \Delta_{k-1}$.
\item\label{item-class-homoscedastic} The $\bar \lambda$-weighted homoscedastic model, which is the $\bar \lambda$-weighted model under the additional assumption that the variances are equal.
\item The $\bar \lambda$-weighted known variance model with known weights and variances.
\end{enumerate}

We wish to do parameter recovery for these classes of Gaussian mixture models, that is, we would like to solve the following problem.

\begin{problem}\label{problem:densityestimationforGMM}
Given $y_1,\ldots, y_N$ sampled from a mixture of $k$ Gaussian densities, recover the parameters $\mu_\ell, \sigma_\ell^2, \lambda_\ell$ for $\ell\in [k]$. 
\end{problem}

These samples can be summarized by sample moments. Motivated by \Cref{problem:densityestimationforGMM} and the \emph{method of moments}, in this paper we solve the following problem.

\begin{problem}\label{problem:moments_populaton}
Given moments from a 
Gaussian $k$-mixture density,
recover the parameters of each Gaussian component.
\end{problem}

The fact that one can recover the mean, variance and weight of each component, up to relabeling, if given the mixture density function
is known as \emph{mixture identifiability} and was proven for univariate Gaussians in 
\cite{teicher1963identifiability}. Multivariate Gaussians are still identifiable \cite{yakowitz1968identifiability} and there exist finitely many moments that identify them~\cite{BS2015}. In this paper we are interested in identifying the means and variances of each mixture from its moments up to a certain order. 

It is important to  distinguish parameter recovery from
density estimation.
For density estimation, one aims to estimate a density that is close to the true mixture density, with no restriction on how close each individual component is.
In this paper we wish to do parameter recovery. Namely, we wish to recover accurate estimates of the mean, variance and mixing weight of each component. It is clear that density estimation follows once all of the parameters are known.
The distinction between density estimation and parameter recovery is illustrated~next.

\begin{example}\label{ex:eight-moments}
Consider the vector of eight moments 
$(m_1,\ldots,m_8)$:
\[ ( 0.1661, \ 2.133, \ 1.3785, \ 12.8629, \ 16.0203, \ 125.6864, \ 239.2856, \ 1695.5639). \]

\begin{figure}[htbp!]
    \centering
    \includegraphics[width = 0.3\textwidth]{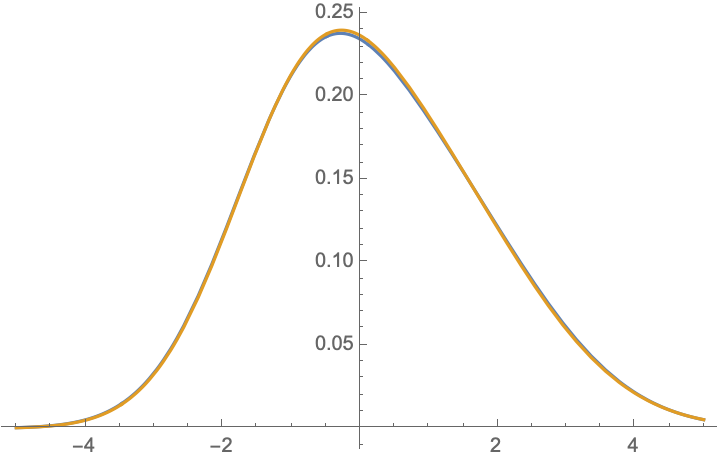}
    \caption{Two distinct Gaussian mixture densities with $k = 3$ components and the same first eight moments.}
    \label{fig:param_recov_ex}
\end{figure}

There are two Gaussian $3$-mixture densities with these eight moments. 
The weights and individual components for each of the mixture models are 
shown in \Cref{fig:individual_components}.
\begin{figure}[htpb!]
    \centering
    \includegraphics[width = 0.35\textwidth]{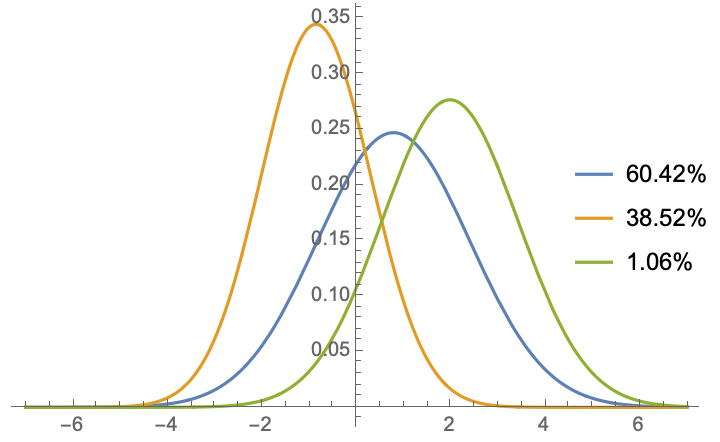}
    \includegraphics[width = 0.35\textwidth]{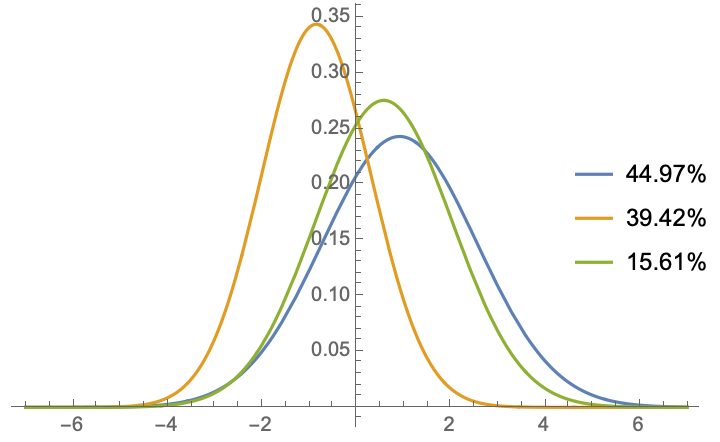}
    \caption{Individual components of two Gaussian mixture models with similar mixture densities.}
    \label{fig:individual_components}
\end{figure}
These densities are shown in \Cref{fig:param_recov_ex} where it is seen that they are almost indistinguishable. In contrast, the individual components and weights are noticeably different. The weights and individual components for each of the mixture models are 
shown in \Cref{fig:individual_components}.

\end{example}

One idea to solve
\Cref{problem:densityestimationforGMM} is to use \emph{maximum likelihood estimation}. Maximum likelihood estimation aims to maximize the likelihood function by solving the following optimization~problem:
\begin{align}
   \text{argmax}_{\mu, \sigma^2, \lambda} \ \  \prod_{i=1}^N  \sum_{i=1}^k \lambda_i \frac{1}{\sqrt{2 \pi \sigma_i^2}} \exp \Big( - \frac{(y_i - \mu_i)^2}{2\sigma_i^2} \Big).\label{mle}
\end{align}

Unless $k=1$, \eqref{mle} is a nonconvex optimization problem and obtaining a global optimum is difficult or impossible. In addition, the maximum likelihood estimator is not always consistent for Gaussian mixture models \cite[Section 3.2]{chen2017consistency}. In general \eqref{mle} is unbounded, so no global maximum exists. Iterative algorithms such as expectation maximization (EM) try to find the largest local maximum  \cite{dempster1977maximum}. On top of being sensitive to the starting point, another downside of the EM algorithm is that it needs to access all data in each iteration, which is prohibitive for applications with large data sets. 
Further, there is no 
 known bound on the number of critical points of the likelihood function, and it has been shown that no such bound can exist that is independent of the sample size \cite{amendola2015maximum}.

Recent work has analyzed the local behavior of \eqref{mle} by considering maximum likelihood estimation for two Gaussians in $\mathbb{R}^n$ when $\lambda_1 = \lambda_2 = \frac{1}{2}$ and all of the covariance matrices are known and equal
--- this is a special case of the models we consider.
It has been shown that in this regime the EM algorithm converges geometrically to the global optimum \cite{daskalakis2017ten,xu2016global}.
Other work has studied the global landscape of the EM algorithm and the structure of local optima in this setting  \cite{chen2020likelihood,xu2016global}. Further work has considered inference for Gaussian mixture models with known mixing coefficients and identity covariance matrices \cite{mei2018landscape} and clustering analysis of the mixture of two Gaussians where the covariance matrices are equal and unknown \cite{cai2019chime}. When these covariance matrices are further assumed to be spherical, \cite{sanjeev2001learning} gives polynomial time approximation schemes for \eqref{mle}
so long as the mixtures are sufficiently well-separated.
Recently, techniques from numerical algebraic geometry have been used to identify the number of components in a Gaussian mixture model \cite{SAG2020}. Further progress has been made on giving optimal sampling bounds needed to learn a Gaussian mixture model \cite{ashtiani2020near}.

Another idea for density estimation in this set up is to use the \emph{generalized method of moments}. The generalized method of moments was proposed in \cite{hansen1982large} and aims to minimize the difference between the fitted moments and the sample moments. For Gaussian mixture models, this again cannot be solved in a way guaranteeing global optimality due to the nonconvexity of the moment equations. Recently this method has been remedied for Gaussian mixture models in one dimension with the same variance parameter, where the authors provably and efficiently find the global optimum of the generalized method of moments \cite{wu2020optimal}. It is important to note that in many of the cases above, assumptions are made on the values that each Gaussian component can take.

In this paper we propose using the method of moments to estimate the density arising from the mixture of $k$ multivariate Gaussians. This methodology was first proposed and resolved for the mixture of two univariate Gaussians by Karl Pearson~\cite{pearson1894contributions}. 
Pearson reduced solving this system of $6$ polynomial equations in the $6$ unknown density parameters, $\mu_\ell, \sigma_\ell^2, \lambda_\ell$, $\ell = 1,2$, to understanding the roots of a degree nine polynomial with coefficients in the sample moments $\overline{m}_i$, $i=1,\dots,5$. 

The method of moments for Gaussian mixture models was revisited in 2010 in a series of papers \cite{kalai2010efficiently,moitra2010settling}. The case of a $k = 2$ mixture model in $n$ dimensions was handled in \cite{kalai2010efficiently} where a polynomial time algorithm was presented. This approach was generalized in \cite{moitra2010settling} where an algorithm for a general $k$-mixture model in $n$ dimensional space was presented that scales polynomially in $n$ and the number of samples required scales polynomially in the desired accuracy. When $k=2,3,4$ the number of solutions to the corresponding moment system are given in \cite{pearson1894contributions,amendola2016moment,amendola2016solving} respectively.

\subsection{Contributions}
The main contribution of this paper is to present new identifiability results for classes of Gaussian mixture models. In \Cref{thm:improv1} we show that a generic Gaussian mixture model is uniquely identifiable (up to label swapping) 
from its first $3k+2$ moments, improving the existing bound of $4k-2$. We then consider families of Gaussian mixture models where some of the parameters are assumed known. 
In \Cref{sec:one_d_moment_varieties} we study the moment equations through the lens of polyhedral geometry and computational algebra. We use this perspective to give algebraic identifiability (\Cref{def:algebraic}) results for three important classes of Gaussian mixture models (\Cref{lem:algebraic_identifiability_lambdas_known}, \Cref{prop:lambda_known_sigma_equal_alg_ident}, \Cref{cor:identifiability_means_unknown}). 

Motivated to perform parameter recovery, we extend these identifiability results by bounding the number of solutions to the moment equations (\Cref{thm:degree_lambdas_unknown}, \Cref{thm:means_unknown_sigma_equal} and \Cref{thm:degree_mus_unknown}). The proofs of these results lead to a homotopy continuation algorithm that implements the method of moments for Gaussian mixture models. The results in the univariate setting, lead to a homotopy method for Gaussian mixture models in $\mathbb{R}^n$ that tracks $\mathcal{O}(n)$ paths. This is because the moment equations for one dimensional Gaussian mixture models are a subsystem of the moment equations for multidimensional Gaussian mixture models. We compare our algorithm with EM and show that when the sample size is large, our algorithm returns solutions that are competitive with EM in a fraction of the time (\Cref{tab:mom_v_EM}).

\section{Preliminaries}\label{sec:preliminaries}
We first introduce some necessary concepts from statistics and algebraic geometry.

\subsection{Method of moments}
This paper focuses on 
an approach for parameter recovery known as the \emph{method of moments}. The method of moments for parameter estimation is based on the law of large numbers. This approach expresses the moments of a density as functions of its parameters.

For $r \geq 0$, we denote the $r$-th 
moment of  a random variable $X$ as
$m_r = \mathbb{E}[X^{r}] $. We consider a statistical model with $n$ unknown parameters, $\theta = (\theta_1,\ldots, \theta_n)$, and consider the moments up to order $n$ as functions of $\theta$, $g_1(\theta),\dots,g_n(\theta)$. For random samples $y_1,\ldots, y_N$ we denote the sample moment by $ \overline{m}_r = \frac{1}{N} \sum_{j=1}^N y_j^r.$

The method of moments works by using samples from the statistical model to calculate sample moments $\overline{m}_1,\ldots, \overline{m}_n$,  and then solve the corresponding system $\overline{m}_i = g_i(\theta)$, $i = 1, \ldots , n$, for the parameters $\theta_1,\ldots, \theta_n$. 

The moments of Gaussian distributions are polynomials in the 
two indeterminants 
$\mu, \sigma^2$ and can be calculated recursively as $M_0(\mu, \sigma^2) = 1,  M_1(\mu, \sigma^2) = \mu$ and
\begin{equation}
\begin{aligned}
    M_i(\mu, \sigma^2) &= \mu \cdot M_{i - 1}(\mu,\sigma^2) + (i - 1) \cdot \sigma^2 \cdot M_{i -2}(\mu,\sigma^2),\quad i\geq 2. \label{momentdef1}
\end{aligned}
\end{equation}
 We calculate the $i$th moment of a mixture of $k$ Gaussian densities as the convex combinations of $M_i(\mu_1, \sigma_1^2), \ldots, M_i(\mu_k, \sigma_k^2)$.

For each non-negative integer $i$, 
we define the polynomial $f_i^k$  
in the $3k$ indeterminants unknowns $\mu_\ell,\sigma_\ell,\lambda_\ell,\, \ell\in [k]$ and sample moment $\overline{m}_i$ to be 
\begin{align}\label{f_polynomial_moment}
f_i^k(\boldsymbol{\mu},\boldsymbol{\sigma^2},\boldsymbol{\lambda}):= 
\lambda_1 M_i(\mu_1,\sigma_1^2)+\cdots +
\lambda_k M_i( \mu_k,  \sigma_k^2) - \overline{m}_i. 
\end{align}
We are then interested in solving the system
\begin{align}\label{moment_polynomials}
f_i^k(\boldsymbol{\mu},\boldsymbol{\sigma^2},\boldsymbol{\lambda}) 
= 0 
\end{align}
under assumptions on the parameters, 
$\mu, \sigma^2,\lambda$ where $i$ varies over an appropriate index set.

We use an overline, e.g., $\overline x$ to indicate a quantity is fixed and known. For the three classes of Gaussian mixture models defined in the introduction, we assume $\overline \lambda \in \Delta_{k-1}$ 
and we solve these respective systems:

\begin{align} 
    & f_i^k(\boldsymbol{\mu},\boldsymbol{\sigma^2},\boldsymbol{\overline \lambda})=0,
    & i\in[2k],
    \quad
    & \text{(\Cref{sec:lambda-weighted})} \label{eq:model1}
    \\
    & f_i^k(\boldsymbol{\mu},\boldsymbol{\sigma^2},\boldsymbol{\overline \lambda})=0,
    &  i\in[k+1], \quad \sigma_1^2=\sigma_2^2=\dots =\sigma_k^2,
    \quad
    & \text{(\Cref{sec:lambda-weighted-homoscedastic})} \label{eq:model2}
    \\
    & f_{i}^k(\boldsymbol{\mu},\boldsymbol{\overline \sigma^2},\boldsymbol{\overline \lambda})=0,
    & i\in[k],
    \quad&\text{(\Cref{sec:means_unknown})}. \label{eq:model3}
\end{align}

\subsection{Genericity}
The concept of generic behavior is prevalent in applied algebraic geometry.
It allows one to make mathematically precise statements without characterizing complicated algebraic sets. 
One says a property is \emph{generic} if it holds on the complement of a proper algebraic subset. In particular, this algebraic subset is lower dimensional, so that non-generic behavior is restricted to a Lebesgue measure-zero set.  
We now define what genericity means in our context.

Consider a statistical model $\MM$ parameterized by $\Theta$, 
where $\theta\mapsto \ff_\theta\in \MM$
and
$\dim(\Theta) = d$.
Let 
\begin{equation}\label{eq:truncated-moment-map}
\phi:\Theta\to \mathbb{R}^d
\end{equation}
be a \emph{truncated moment map} that takes $\theta$ to a list of $d$ moments of the distribution  $\ff_\theta$ of prescribed order.
We assume throughout the article that $\phi$ is a polynomial map.

We define the set of \emph{generic parameters} in $\Theta$ 
(with respect to $\phi$) to be 
\begin{equation}\label{eq:generic-parameters}
\Theta_{\phi,\generic} :=
\{
\theta\in \Theta : \det \Jac(\phi)|_\theta \neq 0
\}.
\end{equation}
The set $\Theta_{\phi,\generic}$ is the complement of an algebraic set restricted to $\Theta$.  Therefore,
when $\Theta_{\phi,\generic}$ is nonempty, 
we have
$\Theta\setminus\Theta_{\phi,\generic}$ is a Lebesgue measure-zero set.

Our next goal is to define generic moments. To do so, first
we define the set $\mathcal{B}_\phi$ in the ambient space of moments to be the smallest 
 algebraic set in $\mathbb{C}^d$ containing 
\begin{equation}\label{eq:branch-locus}
    \{ \phi(\theta)\in \mathbb{R}^d : 
        \det \Jac(\phi)|_\theta =0
    \}
    \subset  \mathbb{C}^d.
\end{equation}
The set of \emph{generic moments} (with respect to $\phi$) is taken to be the set complement
\begin{equation}\label{eq:generic-moments}
\mathbb{C}^d\setminus \mathcal{B}_\phi,
\end{equation}
while the set of \emph{realizable generic moments} is 
\begin{equation}\label{eq:realizable}
\VV_{\phi} := 
\{ 
(m_1,\dots,m_d)\in \phi(\Theta) 
: (m_1,\dots,m_d)\not\in \mathcal{B}_\phi
\}.
\end{equation}

The set 
$\CC^d\setminus \mathcal{B}_\phi$
is the complement of an algebraic set, which allows us to leverage tools from algebraic geometry. 
Proving statements about $\CC^d\setminus \mathcal{B}_\phi$ allows us to make statements relevant for parameter recovery by restricting to realizable generic moments $\VV_\phi$.

\begin{example}\label{ex:k2}
Fix $\overline \lambda=(\frac{1}{2},\frac{1}{2})$ in $\Delta_1$.
As our model, we consider the collection of $\overline \lambda$-weighted Gaussian $2$-mixtures with mean $0$ and second moment fixed to be $1$:
\[
\left\lbrace  (\mu_1,\mu_2,\sigma_1^2,\sigma_2^2) 
:\,
\frac{1}{2}\mu_1+
\frac{1}{2}\mu_2 = 0,
\, 
\frac{1}{2}(\mu_1^2 + \sigma_1^2)+
\frac{1}{2}(\mu_2^2 + \sigma_2^2) = 1
\right\rbrace.
\]
The model can be  reparameterized using 
$\mu_1$ and $\sigma_1^2$ because 
\begin{equation}\label{eq:reduction}
    (\mu_2,\sigma_2^2)
=(-\mu_1,  -2\mu_1^2-\sigma^2_1+2).
\end{equation}

Thus, choose $\phi$ to map $(\mu_1,\sigma_1^2)\in \Theta$ to the third and fourth order moments.
To be specific,
$\phi(\mu_1,\sigma^2_1) = 
(m_3,m_4)$ where 
\[
\begin{array}{llc}
     m_3 \quad &= \quad &  
     \frac{1}{2} \mu_1 (\mu_1^2+3\sigma_1^2)
+
\frac{1}{2} \mu_2 (\mu_2^2+3\sigma_2^2)
\\
m_4  \quad & =\quad &\frac{1}{2}
(\mu_1^2(\mu_1^2+6\sigma_1^2)+3\sigma_1^4)
+
\frac{1}{2}
(\mu_2^2(\mu_2^2+6\sigma_2^2)+3 \sigma_2^4).
\end{array}
\]

The set of generic parameters are those where this $2\times 2$ Jacobian matrix is invertible
\[
\begin{bmatrix}
 3( 3\mu_1^2 + \sigma_1^2 - 1) & 3\mu_1\\
 4 \mu_1^3+12 \mu_1(\sigma_1^2-1) & 6(\mu_1^2+\sigma_1^2-1)
\end{bmatrix}.
\]
In other words, non-generic parameters 
are those such that the determinant of the matrix vanishes, i.e., this equation is satisfied
\[7\mu_1^4+6\mu_1^2(\sigma_1^2-1) +3(\sigma_1^2-1)^2=0.\] 
For this model, the set of generic moments
is the complement of the algebraic set 
\[ \mathcal{B}_\phi= \{
(m_3,m_4)\in \CC^2 :
3m_3^4+2(m_4-3)^3 =0
\}.\]
For example, 
the moments $(m_3,m_4)=(0,3)\in \mathcal{B}_\phi$ are not generic.
This makes sense from a statistical viewpoint, since $(m_1,m_2,m_3,m_4)=(0,1{,}0,3)$ are the moments of a single standard Gaussian distribution, and as such it is not an honest Gaussian 2-mixture.

The set of realizable generic moments is the open set 
\[
\VV_\phi
=\{  (m_3,m_4)\in \phi(\Theta) :  3m_3^4+2(m_4-3)^3 \neq 0 \}.
\]
For $(m_3,m_4)\in \VV_\phi$,
the corresponding moment system has six complex solutions. 
\end{example}

\begin{center}
\begin{figure}
\centering
\begin{overpic}[width=15em,trim={0 5cm 20cm 0cm},clip]{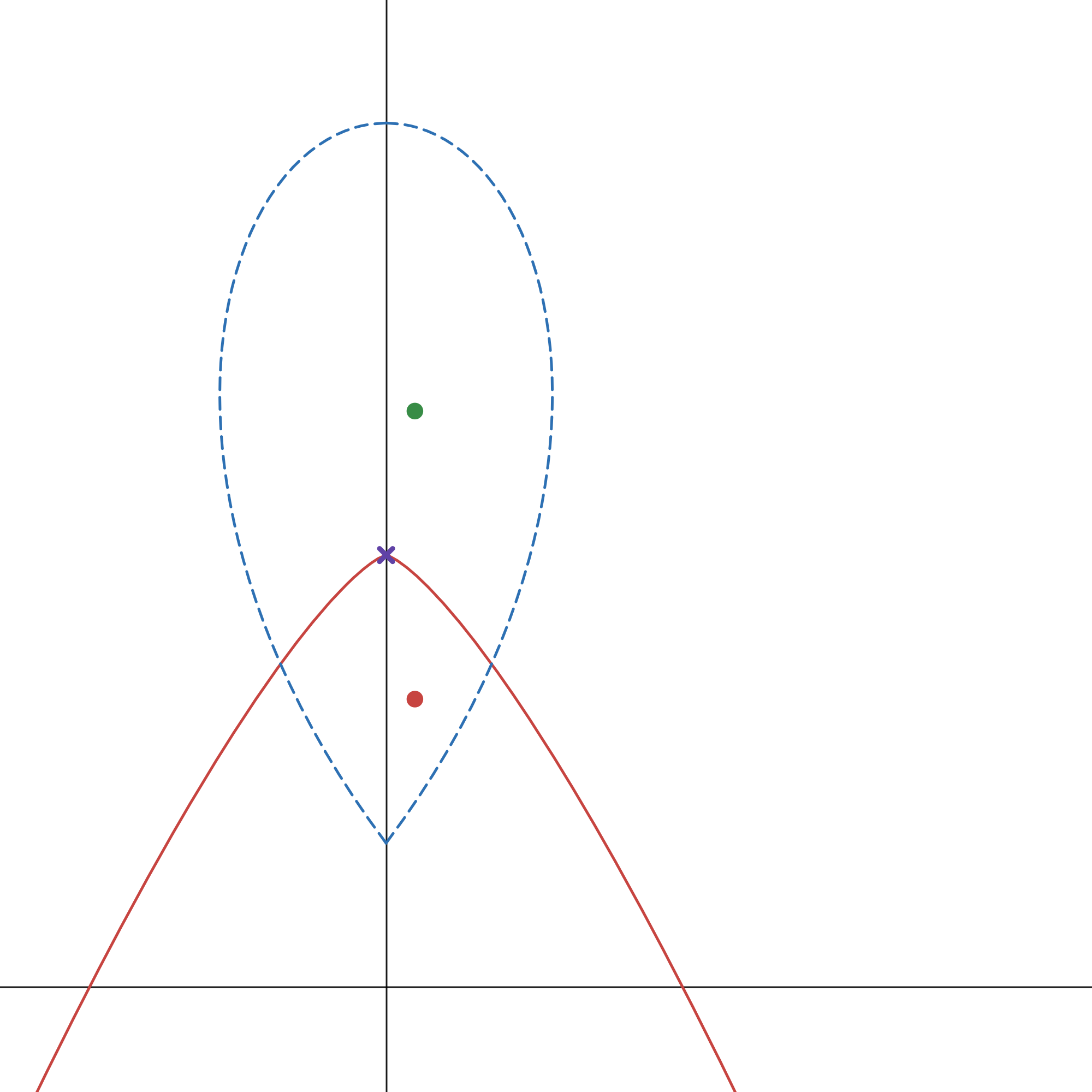} 
\put (80,2) {$\displaystyle m_3$}
\put (40,95) {$\displaystyle m_4$}
\end{overpic}
    \caption{
    The region bounded by the dotted blue curve is the space of realizable moments for this model. 
    The red curve given by 
$3m_3^4+2(m_4-3)^3 =0$ is the set of non-generic moments. The point $(0,3)$ is labeled by a purple cross is non-generic and corresponds to the third and fourth moment of a univariate Gaussian. 
The red point $(\frac{1}{5},2)$ has four real solutions to the moment equations while the green point $(\frac{1}{5},4)$ has two real solutions.
Any choice of moments in the region bounded by the dotted curve and in the complement of the red curve will yield a set of moment equations with six complex solutions, of which two are statistically meaningful (one up to label swapping symmetry).
}
    \label{fig:k2-region-moments}
\end{figure}
\end{center}

\begin{defn}\label{def:algebraic}
A statistical model parameterized by $\Theta\subset \mathbb{R}^d$ is said to be \emph{algebraically identifiable} from $\phi$ 
in \eqref{eq:truncated-moment-map},
if $\Theta_{\phi,\generic}$ is $d$-dimensional.
\end{defn}

In other words, 
for realizable generic moments $\phi(\theta^\ast)$, there are at most finitely many other parameters $\theta$ in $\Theta$  such that $\phi(\theta)=\phi(\theta^*)$.
Similarly, we also consider a notion of identifiability
where $\phi(\theta)=\phi(\theta^*)$ implies $\theta =\theta^*$.

\begin{defn}\label{def:rational}
A statistical model parameterized by $\Theta\subset \mathbb{R}^s$ is \emph{rationally identifiable} from moment $m_1,\dots,m_d$ if the complexified truncated $d$ moments map 
\[ \phi_{\mathbb{C}} \ : \ \mathbb{C}^s \to \mathbb{C}^d \]
is injective on the complement of a proper algebraic subset in $\CC^s$.
\end{defn} 

Computationally, \Cref{def:rational} can be understood as saying that for generic $m$ in the image of $\phi (\Theta)$, the system of equations $\phi(\Theta) = m$ has a unique solution over $\mathbb{C}$.

In this work we focus on Gaussian $k$-mixture models. The moment equations for these models and their submodels can have a label swapping symmetry,  as we saw in \Cref{ex:k2} and \Cref{fig:k2-region-moments}. Thus, when we speak of a Gaussian $k$-mixture model being rationally identifiable, we mean there is a unique solution to the moment equations up to symmetry.

A result in \cite{amendola2018algebraic} states that univariate Gaussian $k$-mixture models are algebraically identifiable from $\phi(\theta)=(m_1,\dots,m_{3k-1})$. This means that for 
generic moments $m_1,\dots,m_{3k-1}$, there are finitely many Gaussian $k$-mixture distributions with exactly those moments.
In particular, the polynomial system \eqref{moment_polynomials} for $i= 0,\ldots,3k-1$ has finitely many solutions. 
Our results in \Cref{sec:one_d_moment_varieties}
are for three classes of submodels of Gaussian $k$-mixtures. 
These are far from trivial as the moments needed to algebraically identify a submodel cannot be immediately determined. Moreover, the number of solutions to the moment systems for the submodel and model can be different, as \Cref{ex:k2} illustrates.

\subsection{Statistically meaningful solutions}

For any set of real-valued sample moments, it is not guaranteed that the moment equations will give 
any statistically meaningful solutions. A \emph{statistically meaningful solution} is a {real valued solution} with positive variances and mixing weights. 
In other words, it is a solution that corresponds to a true density. If the sample moments are inaccurate, it may happen that no solution obtained from the method of moments is statistically meaningful.
By the law of large numbers, as the number of samples goes to infinity the sample moments will converge to the true moments and the method of moments will return a consistent estimator~\cite[Theorem 9.6]{consistentMOM}.

A property of parameterized polynomial systems is that the number of real solutions is constant in open Euclidean sets of the parameter space.
Such open sets can be computed via cylindrical algebraic decomposition \cite{CAD}. The constraints differing real solutions and statistically meaningful ones will further divide these cells. Therefore, in any of these open cells the number of statistically meaningful solutions will be constant. So long as the sample moments lie in a cell that has at least one statistically meaningful solution, which will happen with increasing probability as the sample size increases, the method of moments will return a true density.

\subsection{Sparse polynomial systems}\label{ss:sparsePoly}

We now review concepts from algebraic geometry, but  
for an introduction to sparse polynomial systems see \cite[Ch.~3]{sturmfels2002solving}.
For  each 
$\alpha = (\alpha_1,...,\alpha_n)\in \NN^n$, we have
the monomial $x^\alpha := x_1^{\alpha_1}\cdots x_n^{\alpha_n}$ with exponent $\alpha$. 
A polynomial is a linear combination of monomials. A polynomial is called \emph{sparse} if we know a priori which monomials appear with non-zero coefficients in that polynomial.

Let $\AS_\bullet =(\AS_1,\dots, \AS_k)$ denote a $k$-tuple of nonempty finite subsets of $\NN^n$.
A sparse polynomial system of equations with support $\AS_\bullet$ is given by 
\begin{align*}
\begin{split}
    \sum_{\alpha \in \AS_1} c_{1,\alpha}x^\alpha = \ldots = \sum_{\alpha \in \AS_k} c_{k,\alpha}x^\alpha =0, 
\end{split}
\end{align*}
where $\{c_{i,\alpha } \}_{\alpha\in \AS_i, i\in [k]}$ are the scalar coefficients.  

\begin{rem}
    Observe that we say a sparse polynomial is any polynomial defined by its monomial support. This may be a bit misleading since, in principle, a sparse polynomial could be a linear combination of all $\binom{n+d-1}{n}$ monomials of degree $d$ in $n$ variables. In our case, and as is often the case in applications, the polynomial systems under consideration consist of polynomials that are linear combinations of significantly fewer than $\binom{n+d-1}{n}$ monomials, coinciding with the other common use of `sparse' in the literature.
\end{rem}

The objective in studying sparse polynomial systems is to derive results based off of the monomial support $\mathcal{A}_\bullet$.
The non-trivial extreme case is when each monomial support $\mathcal{A}_i$ has precisely two elements. Each polynomial in the system will have two terms. Such systems are called \emph{binomial systems}. Solving binomial systems where the number of unknowns is equal to the number of equations amounts to computing a Smith or Hermite normal form \cite{CL-binomials}. There are polynomial time algorithms for computing these normal forms \cite{KB1979-smith-hermite}.

\subsection{Mixed volumes}\label{ss:mixedVolume}

Consider two polytopes $P,Q \in \mathbb{R}^n$. The \emph{Minkowski sum} of $P$ and $Q$ is defined as
\[ P + Q := \{ p+q : p \in P, \ q \in Q \}. \]
A \emph{scaling} by $\aaa\in \mathbb{R}_{>0}$ of a polytope $P$ is defined as $\aaa P:= \{\aaa p : p\in P\}$.
Given $s$ convex polytopes in $\mathbb{R}^n$, $K_1,\ldots, K_s$, we consider the standard $n-$dimensional Euclidean volume of a linear combination of these polytopes 
\[
\nu(\aaa_1,\ldots, \aaa_s) = \vol_n\left(\sum_{i=1}^s \aaa_i K_i\right)
\]
where the sum here refers to the Minkowski sum and the  $\aaa_i$ refers to a scaling. 
The polynomial
$\nu(\aaa_1,\ldots, \aaa_s)$ is homogeneous of degree $n$ in $\aaa_1,\ldots, \aaa_s$.

The \emph{mixed volume} of $s$ convex polytopes $K_1,\ldots,K_s$ in $\mathbb{R}^n$ is the coefficient of the $\aaa_1\aaa_2\cdots \aaa_s$ term in $\nu(\aaa_1,\ldots, \aaa_s)$. It is denoted $\mvol(K_1,\ldots, K_s)$. More information on mixed volumes and polytopes can be found in \cite{mvoltext}.

\begin{example}\label{ex:line-segment-det}
The mixed volume of 
$K_1,\dots,K_n$ is easy to describe when 
$K_i$ is a line segment from the origin to the vertex $v_i\in \mathbb{Z}^n$.
The Minkowski sum $\aaa_1 K_1+\cdots +\aaa_n K_n$ is a parallelepiped.
Hence, its volume is given by a determinant:  
\[ \nu ( \aaa_1,\aaa_2,\dots, \aaa_n) = 
\left|
\det 
\begin{bmatrix}
\aaa_1 v_1 & \aaa_2 v_2 &\ldots & \aaa_n v_n
\end{bmatrix}\right|.
\]
In this case, 
$\mvol(K_1,\dots,K_n)$ 
equals the absolute value of the determinant of the matrix with the vertices $v_1,\dots,v_n$ as its columns. 
\end{example}

\subsection{BKK bound}
A celebrated series of papers \cite{bernshtein1979the,khovanskii1978newton,kouchnirenko1976polyedres} gives the connection between sparse polynomial systems described in \Cref{ss:sparsePoly} and polyhedral geometry described in \Cref{ss:mixedVolume}.

Consider a polynomial $f= \sum_{\alpha \in \mathcal{A}} c_\alpha x^\alpha \in \mathbb{R}[x_1,\ldots, x_n]$ with monomial support $\mathcal{A} \subset \mathbb{N}^n$ where $x^\alpha$ denotes $x_1^{\alpha_1}\cdots x_n^{\alpha_n}$. The \emph{Newton polytope} of $f$ is the convex hull of its exponent vectors $\alpha \in \mathcal{A}$, denoted $\newt(f)$. 

\begin{thm}[Bernstein-Khovanskii-Kouchnirenko Bound] \label{bkkbound}
Let $\mathcal{A} :=(\mathcal{A}_1,\ldots , \mathcal{A}_n)$ where $\mathcal{A}_i \subset \mathbb{N}^n$ and $0_n \in \mathcal{A}_i$ for all $i \in [n]$. Let $\mathcal{L}(\mathcal{A})$ be the collection of polynomials $(f_1,\ldots, f_n)$ where $\mathcal{A}_i$ is the support of $f_i$ and $P_i = \conv(\mathcal{A}_i)$ 
for all $i \in [n]$. Consider $F = (f_1,\ldots, f_n)$ where $F \in \mathcal{L}(\mathcal{A})$ and the number of complex solutions to $F = 0$ is finite. Then the number of complex solutions to $F= 0$ is less than or equal to $\mvol(P_1,\ldots, P_n)$. Moreover, for generic choices of $F \in \mathcal{L}(\mathcal{A})$,  equality holds.
\end{thm}

\begin{rem}
The BKK theorem is usually stated without the hypothesis  $0_n \in\mathcal{A}_i$. In this case, the BKK bound counts the number of complex solutions with nonzero coordinates. 
To get a bound on the number of complex solutions, 
the assumption $0_n \in \mathcal{A}_i$ was added in~\cite{TY-li-paper}.
\end{rem}

\begin{rem}
A special case of the BKK bound is the B\'ezout bound \cite[p. 47]{hartshorne_AG}. The B\'ezout bound says that if a polynomial system $F = (f_1,\ldots, f_n)$ has finitely many solutions, then the number of complex solutions with nonzero coordinates is bounded above by $d_1 \cdots d_n$ where $\deg(f_i) =d_i$.
\end{rem}

\Cref{bkkbound} states that for a generic
sparse polynomial system
in  $\mathcal{L}(\mathcal{A})$, the number of complex solutions equals the mixed volume of the system. An obstacle to applying this theorem is that the mixed volume is $\#$P-hard to compute~\cite{mvol_complexity}.

An important property of mixed volumes is that they are \emph{monotonic}. Namely, if $\hat{P}_1 \subseteq P_1$ then
\[ \mvol(\hat{P}_1,P_2,\ldots,P_n) \leq \mvol (P_1,P_2,\ldots,P_n). \]
Therefore by \Cref{ex:line-segment-det}, taking line segments 
$\conv(\{0,v_i\}) = Q_i \subseteq P_i$
for $i \in [n]$, is an easy way to get a lower bound on $\mvol(P_1,\ldots,P_n)$. We would like conditions under which such a lower bound is tight. 

\begin{defn}\cite[Definition VII.32]{mondal2021many}
Let $P_1,\ldots,P_m$ be convex polytopes in $\mathbb{R}^n$. We say that $P_1,\ldots,P_m$ are \emph{dependent} if there exists a nonempty subset $\mathcal{I} \subseteq [m]$ such that $\dim( \sum_{i \in \mathcal{I}}P_i) < |\mathcal{I}|$. Otherwise we say $P_1,\ldots, P_m$ are \emph{independent}. 
\end{defn}
This definition may be difficult to parse on a first read. But it is related to the usual definition of \emph{linear independence}: 
if each $P_i$ is a line through the origin, then the two ideas of dependent agree. 
Moreover, the collection of empty sets is independent.

Given a nonzero vector $ w \in \mathbb{R}^n$ and a convex polytope $P \subseteq \mathbb{R}^n$, we consider the function $P\to \mathbb{R}$, $x\mapsto \langle w,x\rangle$ with  $\langle (w_1,\dots,w_n),(x_1,\dots,x_n) \rangle
:=w_1x_1+\cdots w_nx_n$ denoting the standard inner product.
We denote  
the minimum value  $w$ takes on $P$ by $\val_w(P)$, and the points of $P$ where $w$ is minimized by $\init_w(P)$.
This set of points is a face of $P$, which 
we  call the \emph{face exposed} by $w$. 
 Specifically, we have
\[
\val_w(P) = \min_{x \in P} \ \langle w, x \rangle 
\quad \text{ and }\quad
\init_w(P) = \{x \in P : \langle w,x \rangle \leq \langle w,y \rangle \ \text{ for all } \ y \in P \}.
\]

If $f = \sum_{\alpha \in \newt(f)} c_\alpha x^{\alpha}$, we call
\[
\init_w(f) := \sum_{\alpha \in \init_w(\newt(f))} c_\alpha x^{\alpha}
\]
the \emph{initial polynomial} of $f$. 
For more background on initial polynomials, see
\cite[Chapter~7]{using2005cox}.

\begin{prop} \cite[Proposition VII.39]{mondal2021many} \label{prop:guarantee_maximality}
Let $P_i = \conv(\mathcal{A}_i)$ and $Q_i = \conv(\mathcal{B}_i) \subseteq P_i$ for $i \in [n]$. The following are equivalent:
\begin{enumerate}[leftmargin=*]
         \item \label{item:prop1} $\mvol(P_1,\ldots, P_n) = \mvol(Q_1,\ldots,Q_n) $
            \item\label{item:prop2} One of the following holds:
    \begin{enumerate}
        \item $P_1,\ldots,P_n$ are dependent i.e. $\mvol(P_1,\ldots,P_n) = 0$
        \item\label{item:3b} For each $w \in \mathbb{R}^n \backslash \{0\}$, the collection of polytopes 
        \[\{\init_w (Q_i) : Q_i \cap \init_w(P_i) \, \neq \, \emptyset  \}  \] 
        is dependent. 
    \end{enumerate}
\end{enumerate}
\end{prop}

\Cref{prop:guarantee_maximality} gives conditions under which is suffices to consider the  (potentially much simpler) polytopes $Q_i \subseteq P_i$ to compute the mixed volume of $P_1,\ldots,P_n$.

\begin{example}\label{ex:depPolytopes}
Consider the triangles $P_1 = P_2 = \conv( \{(0{,}0),(1{,}0),(0,1) \})$, 
and the line segments
\[Q_1 = \conv( \{(0{,}0), (1{,}0) \}) \subset P_1, \qquad Q_2 = \conv( \{(0{,}0),(0,1) \}) \subset P_2. \]
Direct computation shows
\[ \mvol(P_1,P_2) = 1 = \mvol(Q_1,Q_2). \]
We can also use \Cref{ex:line-segment-det} to prove $\mvol(Q_1,Q_2) = 1$ and \Cref{prop:guarantee_maximality} to prove $\mvol(P_1,P_2) = \mvol(Q_1,Q_2)$ since for any nonzero $w \in \mathbb{Z}^2$, the collection of polytopes $\{\init_w(Q_i) : Q_i \cap \init_w(P_i) \, \neq \, \emptyset \}$, $i = 1,2$, contains a single point so the collection is dependent.
This type of argument, where we use the dependence of polytopes and \Cref{prop:guarantee_maximality},
is also used in the proofs of \Cref{thm:degree_lambdas_unknown} and \Cref{thm:means_unknown_sigma_equal}.
\end{example}

The following lemma will be of use later to apply \Cref{prop:guarantee_maximality} in the proof of~\Cref{thm:degree_lambdas_unknown}.

\begin{lem}\label{lem:dependent_polytopes}
Let $P_i \subseteq \mathbb{R}^n$ and $Q_i = \conv(\{0_n,v_i\}) \subseteq P_i$ for $i \in [n]$ be convex polytopes. Consider the set, $W$, of nonzero $w \in \mathbb{R}^n$ such that $Q_i \cap \init_w(P_i) \neq \emptyset$ for all $i \in [n]$. I.e.,
\[W = \{w \in \mathbb{R}^n \backslash \{0_n\} : Q_i \cap \init_w(P_i) \neq \emptyset, \ \forall i \in [n] \}. \]
If $\{v_1,\ldots, v_n \}$ are linearly independent, then the collection of polytopes 
\[\{ \init_w(Q_1), \ldots,\init_w(Q_n) \}\] 
are dependent for all $w \in W$.
\end{lem}

\section{Rational Identifiability for Gaussian mixture models}
In this section we make significant progress on an open problem regarding moment identifiability of Gaussian mixtures.

Any univariate Gaussian $k$-mixture model is uniquely identifiable from its first $4k-2$ moments $m_1,\ldots,m_{4k-2}$ \cite{kalai2012disentangling} and  
 algebraically identifiable from the moments $m_1,\ldots,m_{3k-1}$ \cite[Theorem 1]{amendola2018algebraic}.
 The number of moments needed for Gaussian $k$-mixture models to be  rationally identifiable is therefore somewhere between $3k-1$ and $4k-2$.

 We now show that indeed the correct asymptotic order of moments needed for rational identifiability grows like $3k$ and not like $4k$. The proof relies on concepts and recent results from algebraic geometry. We give a concise argument with detailed references.

\begin{rem}
For some historical background, 
note that when $k=2$, the two expressions, $3k$ and $4k-2$, match at six moments. The unique identifiability of a $2$-mixture from $m_1,\dots,m_6$ was implicitly assumed by Pearson in \cite{pearson1894contributions} and first shown to be true by Lazard in \cite{lazard2004injectivity}. 
Regarding rationally identifiable, 
 the number of moments needed for Gaussian $k$-mixtures 
 is conjectured to be precisely $3k$ \cite[Conjecture 3]{amendola2016moment}.
\end{rem}

Let $\mathbb{P}^d$ be the projective space of dimension 
 $d$ whose coordinates
are all $d+1$ moments $m_{i}$
with $i=0,1,\dots, d$.
The main geometric object, first introduced in \cite{amendola2016moment}, is the \emph{Gaussian moment surface} $\mathcal{G}_{1,d}$ in $\mathbb{P}^d$.
It is the projective closure of the image of the parameterization
$\CC^2\to\CC^d$, $(\mu,\sigma^2)\mapsto (m_1,\dots,m_d)$
where $m_i=M_i(\mu,\sigma^2)$ with 
the right-hand side being the bivariate polynomial from \eqref{momentdef1}. 
If we restrict the parameterization to $\mathbb{R}\times \mathbb{R}_{>0}$,
then we have the image consisting of all univariate Gaussian moments up to order $d$. 

The Gaussian moment surface $\mathcal{G}_{1,d}$ is an algebraic set in $\PP^d$, also called a projective variety. 
This means it is
the solution set to a system of homogeneous polynomials in $m_0,\dots,m_d$. 
As an example, for $d=3$,
\[ 
\mathcal{G}_{1,3} = \left\{
(m_0,m_1,m_2,m_3)\in \PP^3 \,:\, 
2\,m_{1}^{3}-3\,m_{0}m_{1}m_{2}+m_{0}^{2}m_{3}=0
\right\}
\]
 is a surface in $\PP^3$, 
 which is smooth everywhere except along the line defined by $m_0=m_1=0$. 
Moreover, by  \cite[Lemma 4]{amendola2018algebraic}, for any $d\geq 3$,
$\mathcal{G}_{1,d}$,
 is singular precisely on the line defined by 
\begin{equation}\label{eq:sing-locus-G1d}
    m_0=m_1=\dots=m_{d-2}=0.
\end{equation}
In addition, since  $\mathcal{G}_{1,d}$ is the closure of the image  of a parameterization by $\CC^2$, 
this projective variety is also irreducible, i.e., if any other closed algebraic sets $A,B$ in $\PP^d$  such that $A\cup B =\mathcal{G}_{1,d}$ 
then either $A$ or $B$ equals $\mathcal{G}_{1,d}$.
Now we are ready to prove a fact about the Gaussian moment surface $\mathcal{G}_{1,d}$.

\begin{lem}\label{lemma:non-degenerate-Gauss}
    The Gaussian moment surface $\mathcal{G}_{1,d}$
    has a non-degenerate Gauss map.
\end{lem}
\begin{proof}
 According to \cite[Section 2]{griffith1979algebraic} and \cite[Theorem 3.4.6]{ivery2003cartan}, if the Gauss map of a projective surface $\ZV$ is degenerate, then $\ZV$ must be one of the following:
\begin{enumerate}[leftmargin=*]
    \item\label{item:degen-linear} a linearly embedded $\mathbb{P}^2$;
    \item\label{item:degen-cone} a cone over a curve;
    \item\label{item:degen-tangential} the tangential variety to a curve.
\end{enumerate}
We are not in case \ref{item:degen-linear} because the degree of  $\mathcal{G}_{1,d}$ is $\binom{d}{2}>1$ for $d>2$  by   
\cite[Corollary 2]{amendola2016moment}. 
In other words $\mathcal{G}_{1,d}$ is a nonlinear surface that cannot be linearly embedded into $\mathbb{P}^2$. 
Case \ref{item:degen-cone} is not possible for  $\mathcal{G}_{1,d}$ either, by the proof of \cite[Theorem 1]{amendola2018algebraic} where great detail is given to exclude the cone over a curve scenario.

It remains to show we are not in  case \ref{item:degen-tangential}. 
If the surface $\mathcal{G}_{1,d}$ is a tangential variety to a curve, 
then either it will contain said curve in its singular locus or the curve will be contained in a plane. 
In the first instance, since the singular locus of $\mathcal{G}_{1,d}$ is a line \eqref{eq:sing-locus-G1d},
we must have $\mathcal{G}_{1,d}$ is the tangential variety to a line. 
This is a contradiction because the tangential variety of a line is the line itself and certainly not a surface. In the second instance, the tangential variety would be a plane, which is case~\ref{item:degen-linear}. 
\end{proof}

Secant varieties have a long history of appearing in applications and more recently in statistics.
For a projective variety $\ZV$,
by definition, the \emph{$k$th secant variety} $\sigma_k(\ZV)$ is the Zariski closure of the union of the $(k-1)$-planes spanned by collections of $k$ points in $\ZV$. 
\begin{example}
Let $\ZV$ be the variety of $m\times n$ rank one matrices in the projective space $\mathbb{P}^{mn-1}$ with coordinates $z_{ij}$ being the $(i,j)$ entry of the matrix.  
For $k\leq \min\{m,n\}$, the secant variety $\sigma_k(\ZV)$ consists of rank at most $k$ matrices. 
\end{example}

For our purposes, we are interested in the $k$th secant variety of the surface $\mathcal{G}_{1,d}$.
 Since the moment of a Gaussian $k$-mixture is a linear combination of the moments of each of the $k$ components, points in $\sigma_k(\mathcal{G}_{1,d})$ represent the moments up to order $d$ of a Gaussian $k$-mixture. 

For $\ZV\subset \mathbb{P}^N$
the dimension of a $k$th secant variety of $\ZV$ satisfies the inequality
 \begin{equation}\label{eq:defective-meaning}\dim(\sigma_k(\ZV))\leq 
 \min\left\{
    k\cdot \dim(\ZV)+k-1, N 
 \right\}.
\end{equation}
When the inequality is strict, we say that $\ZV$ is \emph{$k$-defective}.
In geometric terms, 
a variety $\ZV$ being $k$-defective means that
there are \emph{infinitely} many $(k-1)$-planes passing through any $k$ points in $\ZV$.
On the other hand,
 we say $\ZV$ is \emph{$k$-identifiable} 
if there is a \emph{unique} $(k-1)$-plane that passes through $k$ generic points in $\ZV$. 

In the instance of \eqref{eq:defective-meaning} 
when 
$\ZV= \mathcal{G}_{1,d}$.
By \cite[Theorem 1]{amendola2018algebraic},
 $\mathcal{G}_{1,d}$ is not $k$-defective, i.e., 
 \begin{equation}\label{eq:defective-G1d}\dim(\sigma_k\mathcal{(G}_{1,d})) =  
 \min\left\{
    3k-1, d 
 \right\}.
\end{equation}
From \eqref{eq:defective-G1d} it follows that Gaussian $k$-mixtures are algebraically identifiable from moments up to order $3k-1$.
Furthermore, if one has that
$\mathcal{G}_{1,d}$
is $k$-identifiable,  then we have {rational identifiability} from moments up to order $d$.
We prove this for $d=3k+2$.

\begin{thm}\label{thm:improv1}
Mixtures of $k$ univariate Gaussians are rationally identifiable from moments $m_1,\ldots,m_{3k+2}$.
\end{thm}
\begin{proof}
We wish to apply \cite[Theorem 1.5]{massarenti2022bronowskis}, which states that $k$-identifiability holds for irreducible and non-degenerate varieties $\ZV \subset \mathbb{P}^N$ of dimension $n$ if these  conditions hold:
\begin{enumerate}[leftmargin=*]
    \item\label{case:dim-bound} $(k+1)n + k \leq N$;
    \item\label{case:non-degenerate-Gauss} $\ZV$ has non-degenerate Gauss map;
    \item\label{case:not-defective} $\ZV$ is not $(k+1)$-defective.
\end{enumerate}
We consider $\ZV=\mathcal{G}_{1,3k+2}$ which is $n=2$-dimensional in $\mathbb{P}^{3k+2}$ so that $N=3k+2$. 
The first item holds with equality as 
$(k+1)2+k =  N$. 
Item \ref{case:non-degenerate-Gauss} is the content of \Cref{lemma:non-degenerate-Gauss}.
Thus, it remains to prove item ~\ref{case:not-defective}. This follows from
\eqref{eq:defective-meaning} and ~\eqref{eq:defective-G1d} by substituting $k+1$ for $k$ to~get
\[
\dim(\sigma_{k+1} ( \mathcal{G}_{1,3k+2} ) )
    =
\min\{ 3(k+1)-1,3k+2\}  =3k+2,  
\]
thereby showing 
$\mathcal{G}_{1,3k+2}$ is not $(k+1)$-defective.
Then, by \cite[Theorem 1.5]{massarenti2022bronowskis}, 
we have $\mathcal{G}_{1,3k+2}$ is 
$k$-identifiable.
\end{proof}

\section{Parameter recovery in one dimension}\label{sec:one_d_moment_varieties}
We now present our first set of results for parameter recovery in one dimension: efficiently finding all complex solutions stemming from the moment equations. 
As a consequence of each theorem, we show a model is algebraically identifiable from an appropriately chosen truncated moment map $\phi$.

In each subsection,
we use $\Theta$ to parameterize the model covered in that section. 
In each case $\phi$ (introduced in~\eqref{eq:truncated-moment-map}) is the truncated moment map using the first $\dim(\Theta)$ moments.
\subsection{Mixed volume of \texorpdfstring{$\bar \lambda$-}-weighted models}\label{sec:lambda-weighted}

First consider 
the $\overline \lambda$-weighted model with $k$ mixture components. 
We study the moment system 
\begin{equation}
\begin{aligned} \label{eq:known_lambda}
f_1^k (\boldsymbol{\mu},\boldsymbol{\sigma^2},\boldsymbol{\overline \lambda}) =0,\, \ldots, \,  f_{2k}^k(\boldsymbol{\mu},\boldsymbol{\sigma^2},\boldsymbol{\overline \lambda}) = 0,
\end{aligned}
\end{equation}
where
$\overline \lambda$ are known mixing coefficients,
and $f_i^k$, $i \in [2k]$ is as defined in \eqref{moment_polynomials}. 
In this set up the unknowns are $\mu_\ell, \sigma_\ell$, $\ell \in [k]$.

First, we record the following fact about the moment functions $M_k(\mu, \sigma^2)$ of a Gaussian.

\begin{lem}\label{lem:induction}
The partial derivatives of $M_k$ 
satisfy
\[ \frac{\partial }{\partial \mu} M_k( \mu, \sigma^2) = k \cdot M_{k-1}( \mu, \sigma^2) \quad \text{ and } \quad \frac{\partial }{\partial \sigma^2} M_k( \mu, \sigma^2) = \binom{k}{2} \cdot M_{k-2}( \mu, \sigma^2). \]

\end{lem}

Now we give our first algebraic identifiability result.

\begin{prop}\label{lem:algebraic_identifiability_lambdas_known}
Fix $\overline \lambda\in \Delta_{k-1}$
and let $\Theta$ parameterize 
the $\overline{\lambda}$-weighted  Gaussian $k$-mixture model.
For $(\overline{m}_1,\dots \overline{m}_{2k})\in \VV_\phi$, the set of complex solutions to 
 \eqref{eq:known_lambda} is finite and nonempty. 
\end{prop}

We now use \Cref{bkkbound} and \Cref{prop:guarantee_maximality} to give an upper bound on the number of complex solutions to \eqref{eq:known_lambda}. 
Recall that if $N$ is odd, the double factorial is defined as 
\[ N!! = 1 \cdot 3 \cdot 5 \cdots N. \]

\begin{thm}[Mixing coefficients known]\label{thm:degree_lambdas_unknown}
Fix $\overline \lambda\in \Delta_{k-1}$,
and let $\Theta$ parameterize 
the $\overline{\lambda}$-weighted  Gaussian $k$-mixture model.
For
$(\overline{m}_1,\dots \overline{m}_{2k})\in \VV_\phi$,
 the number of complex solutions to \eqref{eq:known_lambda} is 
 at most $(2k-1)!!\cdot k!$.
\end{thm}

\begin{rem}
An instance of the previous theorem is when the mixing weights are all equal. In this case, there are $(2k-1)!!$ solutions up to the standard label-swapping symmetry.
There exist techniques in numerical nonlinear algebra that exploit this symmetry for computational speed up. One such technique known as \emph{monodromy} was recently used for this problem with success~\cite[Section~4.1]{amendola2016solving}.
\end{rem}

\subsection{Mixed volume of \texorpdfstring{$\bar \lambda$-}-weighted homoscedastic models}\label{sec:lambda-weighted-homoscedastic}

We consider the $\bar\lambda-$weighted homoscedastic model. In this setting the means are unknown and the variances are unknown but all equal. In this case, a $k$-mixture model has $k+1$ unknowns.

We consider the moment system 
\begin{equation}
\begin{aligned} \label{eq:known_lambda_homoscedastic}
f_1^k (\boldsymbol{\mu},\boldsymbol{\sigma^2},\boldsymbol{\overline \lambda}) =0, \, \ldots , \, f_{k+1}^k(\boldsymbol{\mu},\boldsymbol{\sigma^2},\boldsymbol{\overline \lambda}) = 0,
\quad 
\sigma_1^2=\sigma_2^2=\cdots=\sigma_k^2
\end{aligned}
\end{equation}
where $\overline \lambda$ are known mixing coefficients, 
and $f_i^k$, $i \in [k+1]$ is as defined in \eqref{moment_polynomials}. In this set up, the $k+1$ unknowns are $\mu_\ell $, $\ell \in [k]$ and $\sigma^2$.

Again, the first step is to prove that this model is algebraically identifiable.

\begin{prop}\label{prop:lambda_known_sigma_equal_alg_ident}
Fix $\overline \lambda\in \Delta_{k-1}$,
and let $\Theta$ parameterize 
the $\overline{\lambda}$-weighted homoscedastic model.
For
{$(\overline{m}_1,\dots \overline{m}_{k+1})\in \VV_\phi$,
}
 the set of complex solutions to \eqref{eq:known_lambda_homoscedastic} is finite and nonempty. 
\end{prop}

We bound the number of solutions to \eqref{eq:known_lambda_homoscedastic} for a generic $\bar \lambda$-weighted homoscedastic mixture model using mixed volumes.

\begin{thm}[Mixing coefficients known, variances equal]\label{thm:means_unknown_sigma_equal}
Fix $\overline \lambda\in \Delta_{k-1}$,
and let $\Theta$ parameterize
{the $\overline{\lambda}$-weighted  homoscedastic Gaussian $k$-mixture model.}
For
$(\overline{m}_1,\dots \overline{m}_{k+1})\in \VV_\phi$,
 the number of complex solutions to \eqref{eq:known_lambda_homoscedastic} is 
 at most  $\frac{(k+1)!}{2}$.
\end{thm}

\subsection{Finding all solutions using homotopy continuation}\label{sec:finding_all_sols}
To do parameter recovery for any of the set ups described in
\Cref{sec:lambda-weighted} and \Cref{sec:lambda-weighted-homoscedastic}, it is not enough to know the number of complex solutions to the moment equations, we need a way to find all of them. Finding all complex solutions to a zero dimensional variety is a well-studied problem in numerical algebraic geometry. 
We outline the basics of homotopy continuation below but give \cite{bertini-book}
for a detailed reference.

Consider the system of polynomial equations
\begin{align*}
    F(\boldsymbol x) &= \big( f_1( \boldsymbol x), f_2(\boldsymbol x) ,\ldots, f_n(\boldsymbol x) \big) = 0
\end{align*}
where the number of complex solutions to $F(\boldsymbol x) = 0$ is finite and $\boldsymbol x = (x_1,\ldots,x_n)$.
The main idea is to construct a \emph{homotopy} \cite{li1997numerical}
\begin{equation}\label{eq:homotopy}
    H(\boldsymbol x;t) = \gamma(1-t)G(\boldsymbol x) + t F(\boldsymbol x)
\end{equation}
such that these three conditions hold:
\begin{enumerate}
    \item the solutions to $H(\boldsymbol x,0)=G(\boldsymbol x)=0$ are trivial to find, 
    \item\label{item:nosingularities} there are no singularities along the path $t \in [0,1)$, and
    \item all isolated solutions of $H(\boldsymbol x,1)=F(\boldsymbol x) = 0$ can be reached.
    \end{enumerate}
 Using a random $\gamma \in \mathbb{C}$ ensures that 
 there are no singularities along the path $t \in [0,1)$ so condition
 two is met. 
 Continuation methods known as \emph{predictor-corrector methods} are used to track the solutions from $G(\boldsymbol x) = 0$ to $F(\boldsymbol x)=0$ as $t$ varies from $0$ to $1$ \cite{butcher2003numerical}.

It is standard terminology to call $G(\boldsymbol x)=0$  the \emph{start system}. There are many choices for constructing a start system. 
A \emph{total degree} start system is of the form
\[
    G(\boldsymbol x) = \{x_1^{d_1}-1,\ldots, x_n^{d_n}-1 \} =0,
\]
where $d_i$ is the degree of $f_i$. The number of solutions to $G(\boldsymbol x)=0$ is $d_1\cdots d_n$, which is the B\'ezout bound of $G(\boldsymbol x)$. Using this start system, one must track $d_1\cdots d_n$ paths in order to find all solutions to $F(\boldsymbol x) = 0$. If $F(\boldsymbol x)=0$ has close to $d_1\cdots d_n$ solutions this is a reasonable choice.

If $F(\boldsymbol x)=0$ is a polynomial system with fewer solutions than the B\'ezout bound, then tracking $d_1\cdots d_n$ paths is unnecessary. 
Instead, one can leverage the system's Newton polytopes to construct a start system that has as many solutions as the mixed volume.  One example of this is the polyhedral homotopy method  \cite{huber1995a}
which tracks $N$ paths where $N$ is the BKK bound described in \Cref{bkkbound}.
This method constructs several binomial start systems, that is, systems formed by polynomials with precisely two terms.
A start system which tracks $N$ paths where $N$ is the BKK bound is called \emph{BKK optimal}.
The main drawback to polyhedral methods is that there is often a computational bottleneck associated with computing the start system. 
Our related approach circumvents this bottleneck and relies on the following lemma. 

The collection of Newton polytopes of a polynomial system $F = (f_1,\ldots,f_n)$ is denoted by $\newt(F) = (\newt(f_1),\ldots, \newt(f_n))$. 

\begin{lem}\label{lemma:Binomial-BBK-homotopy}
Suppose $G(\boldsymbol x) =0$ is a general sparse binomial system 
such that we have $\mvol(\newt(G))=\mvol(\newt(F))$ and 
$\newt(G)\subseteq\newt(F)$ coordinate-wise. 
If the origin is in each Newton polytope of $G$, 
then the three items above hold for the homotopy 
\eqref{eq:homotopy}.
\end{lem}

The fact that the total degree homotopy works is a special case of the previous lemma applied to polynomials with full monomial support. Combining \Cref{lemma:Binomial-BBK-homotopy} with \Cref{thm:degree_lambdas_unknown} and \Cref{thm:means_unknown_sigma_equal} we get the following corollary.

\begin{cor}\label{thm:optimal_start_system}
The binomial system induced by the polytopes $Q_i$ in the proofs of \Cref{thm:degree_lambdas_unknown} 
and \Cref{thm:means_unknown_sigma_equal} constructs a BKK optimal homotopy continuation start system for the corresponding moment system. 
\end{cor}
\begin{proof}
In this proof we construct the homotopy; give its start points; and show that it is optimal. We only show the details for the case of \Cref{thm:degree_lambdas_unknown} because the other statement's proof is analogous. 

Consider the binomial system 
\begin{equation*}
\begin{aligned}
    g_{2\ell-1} &= a_\ell \mu_{\ell}^{2\ell-1} + b_\ell, \quad \ell\in [k] \\
    g_{2\ell} &= c_{\ell } (\sigma_{\ell}^2)^{\ell} + d_{\ell},\quad \ell\in [k]
    \end{aligned}
\end{equation*}
where $a_\ell,b_\ell, c_\ell, d_\ell \in \mathbb{C}^*$ are generic for  $\ell\in [k]$.

Since each  $g_{2\ell-1}$ and $g_{2\ell}$ is a univariate polynomial in distinct variables,  multiplying the degrees shows that there are $(2k-1)!!k!$ solutions. 
This number agrees with the mixed volume of the respective moment system by \Cref{thm:degree_lambdas_unknown}. 
Moreover, the solutions are the start points of the homotopy 
\begin{equation}\label{eq:homotopy_from_cor}
H(\boldsymbol \mu,\boldsymbol \sigma^2;t) :=\begin{cases}
    (1-t)\gamma g_1+tf_1^k = 0\\
    \vdots\\
    (1-t)\gamma g_{2k}+tf_{2k}^k = 0.
\end{cases}
\end{equation}
where $f_i^k$ are defined as in \eqref{eq:known_lambda}. By \Cref{lemma:Binomial-BBK-homotopy} the result follows. 
\end{proof}

The homotopy in \eqref{eq:homotopy_from_cor} is an example of a polyhedral homotopy~\cite{huber1995a}. Polyhedral homotopies are constructed by multiplying each term of the system by $t$ to some power. When $t=1$ we have our original system, while at the limit $t\to0$ the system decomposes to multiple binomial systems. Determining the binomial systems at the limit is a computational bottleneck for the polyhedral homotopy method.  
\Cref{thm:optimal_start_system} bypasses the computational bottleneck associated with polyhedral homotopy methods. 
Therefore, the proof of each theorem gives the number of complex solutions to the corresponding variety \emph{and} provides an optimal start system that is binomial.
In contrast, the polyhedral homotopy solver in \cite{HomotopyContinuationjulia} failed to solve \eqref{eq:known_lambda} when $k = 6$ because it could not find a start system. On the other hand, we can immediately read off a binomial system using \Cref{thm:degree_lambdas_unknown}.

\begin{example}\label{ex:start_system}
Consider \eqref{eq:known_lambda} when $k = 2$ and $\overline \lambda = (1\slash 2, 1 \slash 2)$. Here we have
\begin{align*}
    f_1^2(\boldsymbol \mu,\boldsymbol \sigma^2,\boldsymbol{\overline \lambda} ) &= \frac{1}{2} \mu_1 + \frac{1}{2} \mu_2 - \overline{m}_1 \\
    f_2^2(\boldsymbol \mu,\boldsymbol \sigma^2,\boldsymbol{\overline \lambda}) &= \frac{1}{2}(\mu_1^2 + \sigma_1^2) + \frac{1}{2}(\mu_2^2 + \sigma_2^2) - \overline{m}_2\\
    f_3^2(\boldsymbol \mu,\boldsymbol \sigma^2,\boldsymbol{\overline \lambda} ) &= \frac{1}{2}(\mu_1^3 + 3 \mu_1 \sigma_1^2) + \frac{1}{2}(\mu_2^3 + 3 \mu_2 \sigma_2^2) - \overline{m}_3\\
    f_4^2(\boldsymbol \mu,\boldsymbol \sigma^2,\boldsymbol{\overline \lambda} ) &= \frac{1}{2}(\mu_1^4 + 6 \mu_1^2 \sigma_1^2 + 3 \sigma_1^4) + \frac{1}{2}(\mu_2^4 + 6 \mu_2^2\sigma_2^2 + 3 \sigma_2^4) - \overline{m}_4
\end{align*}
In this case we consider the start system:
\begin{align*}
    g_1 &=  \mu_1 - 10,  &&g_2 =  \sigma_1^2 - 12, \\
    g_3 &= \mu_2^3 - 27,  &&g_4 = \sigma_2^4 -4.
\end{align*}
This gives six start solutions of the form $(\mu_1,\sigma_1^2, \mu_2, \sigma_2^2)$:

\begin{align*}
    &(10, 12, \eta \cdot 3, 2 ), &&(10, 12,\eta^2 \cdot 3, 2 ), &&(10, 12, 3,2 ) \\
      &(10, 12, \eta \cdot 3, -2 ), &&(10, 12, \eta^2 \cdot 3, -2 ),  &&(10, 12, 3, -2 ) \\
\end{align*}
where $\eta$ is a primitive third root of unity.
We chose integers as the coefficients for ease of exposition. In practice, random complex numbers with norm close to one are used as the coefficients. 
\end{example}

For the $\lambda$-weighted model discussed in \Cref{sec:lambda-weighted}, the B\'ezout bound of \eqref{eq:known_lambda} is $(2k)!$ but \Cref{thm:degree_lambdas_unknown} showed that the mixed volume is $(2k-1)!!k!$. The limit of
\[  \frac{(2k)!}{(2k-1)!!k!} \]
tends to infinity as $k \to \infty$, showing that for large $k$, our start system is significantly better than a total degree one. 

\begin{rem}
For the $\lambda-$weighted homoscedastic model discussed in \Cref{sec:lambda-weighted-homoscedastic}, \Cref{thm:means_unknown_sigma_equal} shows that there are at most $\frac{(k+1)!}{2}$ solutions even though the B\'ezout bound is $(k+1)!$. In this case using our start system, one would track half as many paths as a total degree start system would. 
\end{rem}

\subsection{Means unknown}\label{sec:means_unknown}

A final case of interest is the $\lambda$-weighted, known variance model. This is where only the means are unknown. 
This set up was considered in high dimensions in \cite{chen2020likelihood, daskalakis2017ten, xu2016global}.

We consider the moment system 
\begin{equation}
\begin{aligned} \label{eq:unknown_mean}
f_1^k (\boldsymbol \mu,\boldsymbol{\overline \sigma^2},\boldsymbol{\overline \lambda}) = \cdots = f_{k}^k(\boldsymbol \mu,\boldsymbol{\overline \sigma^2},\boldsymbol{\overline \lambda}) = 0
\end{aligned}
\end{equation}
where 
$\overline \lambda$ are known mixing coefficients, $\overline{\sigma_\ell}^2$ is a known variance,
and $f_\ell^k$, $\ell \in [k]$ is as defined in \eqref{moment_polynomials}. In this set up, the unknowns are $\mu_\ell $, $\ell \in [k]$.

\begin{thm}[Means unknown]\label{thm:degree_mus_unknown}
Fix $\overline \lambda\in \Delta_{k-1}$,
and let $\Theta$ parameterize the $\overline{\lambda}$-weighted  known variance Gaussian $k$-mixture model.
For
$(\overline{m}_1,\dots \overline{m}_{k})\in \VV_\phi$,
 the number of complex solutions to \eqref{eq:unknown_mean} is 
 at most  $k!$ where the inequality is tight.
\end{thm}

\begin{cor}[Equal weights and variances]\label{cor:identifiability_means_unknown}
A Gaussian $k-$mixture model with uniform mixing coefficients and known and equal variances is identifiable up to label-swapping symmetry using moments $(m_1,\ldots,m_k)$ {$\in \VV_\phi$}. 
\end{cor}

As a consequence of \Cref{thm:degree_mus_unknown}, when only the means are unknown the B\'ezout bound is equal to the BKK bound. In this case using polyhedral homotopy gives no advantage to total degree. 

As discussed in \Cref{cor:identifiability_means_unknown}, when the mixture weights and variances are known and equal the standard label-swapping symmetry observed with mixture models gives only one solution up to symmetry. Tracking a single path from a total degree start system, this one solution is easy to find. 

\begin{example}
When $k=2$, $\lambda_1 = \lambda_2 = \frac{1}{2}$ and $\sigma_1^2 = \sigma_2^2 = \sigma^2$ is a known parameter, we can symbolically solve the corresponding moment system and see that up to symmetry
\[\mu_1 = \overline{m}_1 - \sqrt{-\overline{m}_1^2 + \overline{m}_2 - \sigma^2}, \qquad \mu_2 = \overline{m}_1 + \sqrt{-\overline{m}_1^2 + \overline{m}_2 - \sigma^2}. \]
This shows that so long as $-\overline{m}_1^2 + \overline{m}_2 - \sigma^2>0$ we are guaranteed to get something statistically meaningful. A picture of that region in $\mathbb{R}^3$ is shown in  \Cref{fig:stat_sig_Cell}.

\begin{figure}[htbp!]
    \centering
    \includegraphics[width = 0.3\textwidth]{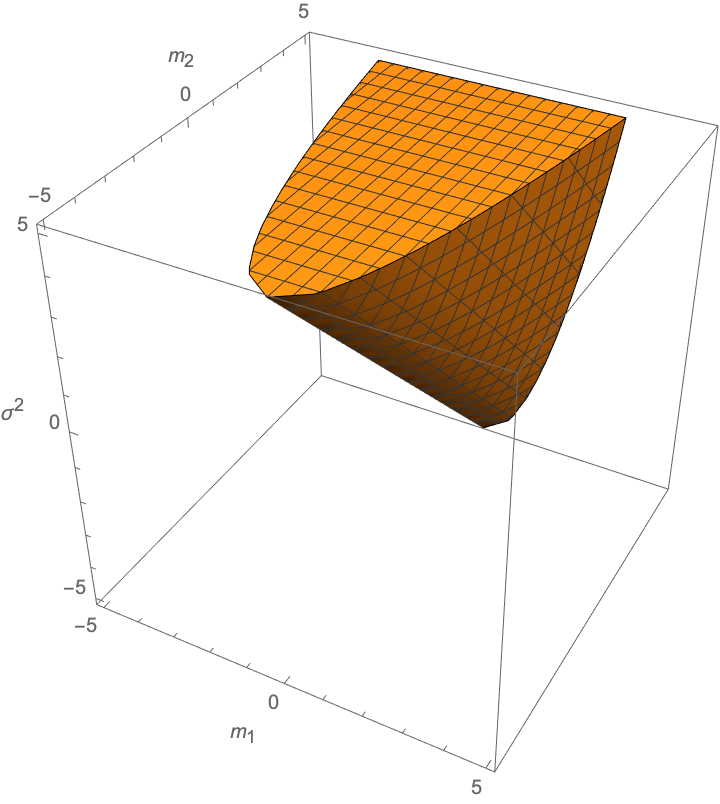}
    \caption{Region in the space of parameters $\overline{m}_1,\overline{m}_2,\sigma^2$ where there are statistically meaningful solutions for $k=2$ mixture model with unknown means and $\lambda_1 = \lambda_2 = \frac{1}{2}$.} 
    \label{fig:stat_sig_Cell}
\end{figure}
\end{example}

\section{Parameter recovery for high dimensional Gaussian mixture models}\label{sec:highdimstatistics}

\Cref{sec:one_d_moment_varieties} gives upper bounds on the number of solutions to the moment equations for Gaussian mixture models where some parameters of the model are assumed known. Using homotopy continuation algorithms 
we can efficiently perform density estimation in these cases. 
We now use our results to do density estimation on Gaussian mixture models in high dimensions.

\subsection{High-dimensional Gaussian mixture models}

To set notation, recall a random variable $X \in \mathbb{R}^n$ is distributed as a \emph{multivariate Gaussian} with mean $\mu \in \mathbb{R}^n$ and symmetric positive definite covariance matrix
$\Sigma \in \mathbb{R}^{n \times n}$, 
if it has density
\begin{align*}
    f_X(x_1,\ldots,x_n | \mu, \Sigma ) &= ((2 \pi)^n \det(\Sigma) )^{- 1 \slash 2} \exp \left( - \frac{1}{2}(x - \mu )^T  \Sigma^{-1} (x - \mu) \right).
\end{align*}
We denote $X \sim \mathcal{N}(\mu, \Sigma)$.

A random variable is distributed as the \emph{mixture of $k$ multivariate Gaussians} if it 
is the convex combination of $k$ multivariate Gaussian densities. It has probability density function
\begin{align*}
    f_X(x_1,\ldots,x_n | \lambda_\ell, \mu_\ell, \Sigma_\ell)_{\ell = 1,\ldots, k} &=
    \sum_{\ell=1}^k \lambda_\ell 
        f_{X_\ell}(x_1,\ldots,x_n \vert \mu_\ell, \Sigma_\ell )
 \end{align*}
where $(\lambda_1,\ldots,\lambda_k) \in \Delta_{k-1}$, $\mu_\ell \in \mathbb{R}^n$, and $\Sigma_\ell \in \mathbb{R}^{n \times n}$ is symmetric and positive definite for $\ell \in [k]$. Here we write $X \sim \sum_{\ell=1}^k \lambda_\ell\mathcal{N}(\mu_\ell, \Sigma_\ell).$

Let $f_X: \mathbb{R}^n \to \mathbb{R}$ be the probability density function of a random vector $X = (X_1,\ldots, X_n)$. 
We denote,
the \emph{$(i_1,\ldots, i_n)-$th moment} of $X$ by 
\begin{align*}
    m_{i_1,\ldots,i_n} &= \mathbb{E}[X_1^{i_1}\cdots X_n^{i_n}] = \int_{\mathbb{R}^n} x_1^{i_1} \cdots x_{n}^{i_n} f_X(x_1,\ldots,x_n) d x_1 \cdots d x_n
\end{align*}
where $i_s \geq 0$ for all $s \in [n]$ and the non-negative integer $i_1 + \ldots + i_n = d$ is the \emph{order} of $m_{i_1,\ldots,i_n}$.

We express $m_{i_1,\ldots,i_n}$ as a polynomial in the parameters
of a Gaussian mixture using the moment generating function. 
These are given by the identity:
\begin{align}
    \sum_{i_1,\ldots,i_n \geq 0} \frac{m_{i_1,\ldots, i_n}}{i_1! \cdots i_n!} t_1^{i_1} \cdots t_n^{i_n} &= \sum_{\ell = 1}^k \lambda_\ell \exp ( t_1 \mu_{\ell 1} + \ldots t_n \mu_{\ell n}) \cdot \exp( \frac{1}{2} \sum_{i,j = 1}^n \sigma_{\ell i j} t_i t_j). \label{moment_gen_GMM}
\end{align}

Using Taylor's formula, we can expand the left 
side of \eqref{moment_gen_GMM} and equate coefficients of each side to get $m_{i_1,\ldots,i_n}$. 
As in the univariate case, 
\begin{equation}\label{eq:m-to-M}
m_{i_1,\ldots,i_n} = \lambda_1 M_{i_1,\ldots,i_n}(\mu_1,\Sigma_1) + \ldots + \lambda_kM_{i_1,\ldots,i_n}(\mu_k,\Sigma_k)      
\end{equation}
where $M_{i_1,\ldots,i_n}(\mu,\Sigma)$ denotes the $(i_1,\ldots,i_n)-$th moment of a multivariate Gaussian random variable $X \sim \mathcal{N}(\mu,\Sigma)$.

The $m_{0{,}0,\ldots,0,i_s,0,\ldots 0}-$th moment is the same as the $i_s-$th order moment for the univariate Gaussian mixture model $\sum_{\ell =1}^k \lambda_\ell \mathcal{N}(\mu_{\ell s}, \sigma_{\ell s s}) $. 
This observation is key 
to our proposed density estimation algorithm in \Cref{sec:algorithm} and it follows from the property that marginal distributions of a Gaussian are Gaussian themselves.

\subsection{Dimension reduction and parameter recovery algorithm}\label{sec:algorithm}

We propose an algorithm for {parameter} estimation of multivariate Gaussian densities using the method of moments. The main idea is that if we use higher order moment equations, parameter estimation for multivariate Gaussian mixture models reduces to multiple instances of density estimation for univariate Gaussian mixture models. 
In this subsection we present
\Cref{algorithm:multivariate_density_estimation} to describe how to recover the parameters of a $\overline \lambda$-weighted multivariate Gaussian.

\begin{algorithm}[hbt!]
\DontPrintSemicolon 
\KwIn{
Mixing coefficients $\overline  \lambda_\ell$, $\ell \in [k]$ and moments: 
\begin{align*}
\overline{\mathfrak{m}}_{1,i} &:= \left\{\overline{m}_{e_i},\ldots, \overline{m}_{De_i}  \right\} \quad i\in [n]  \quad \text{ and }\quad D=3k+2.
\\
\overline{\mathfrak{m}}_2 &:= \{\overline{m}_{e_i + e_j}, \overline{m}_{2 e_i + e_j},\ldots,\overline{m}_{k e_i + e_j} \ :\  i, j \in [n], \ i \neq j  \}
\end{align*}
from a multivariate Gaussian mixture model:
\[ \bar \lambda_1 \mathcal{N}(\mu_1, \Sigma_1) + \cdots + \bar \lambda_k \mathcal{N}(\mu_k, \Sigma_k). \]
}
\KwOut{
    Parameters 
    $\mu_\ell  \in \mathbb{R}^n$, $\Sigma_\ell \succ 0$ for $\ell \in [k]$ such that $\overline{\mathfrak{m}}_{1,i}, \overline{\mathfrak{m}}_2$ are the moments of
    $\sum_{\ell=1}^k \bar \lambda_\ell \mathcal{N}(\mu_\ell, \Sigma_\ell)$. 
}
  \For{$i = 1,\dots,n$\label{alg:step-for-loop}}{
    With mixing coefficients $\bar \lambda_\ell$ and
  moments
    $\overline{\mathfrak{m}}_{1,i}$,
    find all solutions to \eqref{eq:known_lambda} with $ \mu_{\ell i} \in \mathbb{R}$ and $ \sigma_{\ell i i } > 0$ for all $\ell \in [k]$. 
    \label{alg:step3}
    \;
    Select the solution consistent with the moments up to $\overline{m}_{De_i}$.
    \label{alg:step-select} \;
}
    Form a linear system of $k \binom{n}{2}$ equations in 
    $k \binom{n}{2}$ unknowns
    by substituting $\overline{\lambda}_{\ell}, \mu_{\ell i}$ and $\sigma_{\ell i i}$ into $m_{t e_i + e_j}=\overline{m}_{t e_i + e_j}$ for $i,j \in [n]$, $t \in [k]$.\label{alg-step:form-system}
    \;      
    Solve the linear system for $\sigma_{\ell ij}, i\neq j$. \label{alg:linear}\;    
        {Return $ \mu_\ell, \Sigma_\ell$, $\ell\in [k]$.}
        \label{alg:step6} 
        
\caption{
Parameter recovery
for Mixtures of $\bar \lambda$-Weighted Multivariate Gaussians}\label{algorithm:multivariate_density_estimation}
\end{algorithm}

   \begin{rem}\label{rem:linear-system}
In Step~\ref{alg:linear}
of \Cref{algorithm:multivariate_density_estimation}
we solve a linear system. The fact that this is a linear system relies on the choice of moments in $\mathfrak{\overline m}_2$.
If we use a different set of moments, the system could become nonlinear or fail to uniquely recover the off diagonal entries of $\Sigma_\ell$. 

 \end{rem}

\begin{example}\label{ex:two-mix}
Suppose $X \sim 0.25 \cdot  \mathcal{N}(\mu_1, \Sigma_1) + 0.75\cdot \mathcal{N}(\mu_2, \Sigma_2)$.

The first step of \Cref{algorithm:multivariate_density_estimation}
use the moments $\overline{\mathfrak{m}}_{1,1}$. 
Here we use the input

\begin{align*}
  \overline{\mathfrak{m}}_{1,1} &= [\overline m_{10}, \dots, \overline m_{80}] = [ -0.25,\  2.75,\ -1.0,\  22.75,\  -6.5, \ 322.75, \ -58, \ 6569.75] \\
  \overline{\mathfrak{m}}_{1,2} &= [\overline m_{01},\dots,\overline m_{08}] = [2.5, \ 16.1, \ 74.5, \ 490.6, \ 2921.3, \ 20595.3, \ 142354, \ 1080135.4 ]. 
\end{align*}

In Step~\ref{alg:step3},
we solve \eqref{eq:known_lambda} with $\overline{\mathfrak{m}}_{1,1}$ for the unknowns $\mu_{11},\mu_{21},\sigma_{111},\sigma_{211}$:
\begin{align*}
        -0.25 &= 0.25 \cdot \mu_{11} + 0.75 \cdot \mu_{21} \\ 
           2.75 &= 0.25 \cdot(\mu_{11}^2 + \sigma_{111}) + 0.75 \cdot(\mu_{21}^2 + \sigma_{211})\\
          -1 &= 0.25 \cdot(\mu_{11}^3 + 3 \mu_{11} \sigma_{111}) + 0.75 \cdot (\mu_{21}^3 + 3 \mu_{21} \sigma_{211})\\
           22.75 &= 0.25 \cdot ( \mu_{11}^4 + 6 \mu_{11}^2 \sigma_{111} + 3 \sigma_{111}^2) + 0.75 \cdot ( \mu_{21}^4 + 6 \mu_{21}^2 \sigma_{211} + 3 \sigma_{211}^2).
\end{align*}
We find that there are two statistically meaningful solutions:
\begin{align*}
            (\mu_{11}, \mu_{21}, \sigma_{111},\sigma_{211}) &= (-1, 0,\ 1,\ 3) \\
            (\mu_{11}, \mu_{21}, \sigma_{111},\sigma_{211}) &= (0.978,\  -0.659,\ 2.690,\ 2.017).
\end{align*}
The first solutions has moments agreeing with the input $\overline{\mathfrak{m}}_{1,1}$,
while the second takes these values 
$[-1.5, 316.6, 168.7, 6243.9]$. Therefore, we select the first solution.
Proceeding with the second iteration of the for loop in Step~\ref{alg:step-for-loop}, we solve for the unknowns $\mu_{12},\mu_{22},\sigma_{122},\sigma_{222}$ 
using the sample moments $\overline{\mathfrak{m}}_{1,2}$ to obtain
    \begin{align*}
        (\mu_{12}, \mu_{22}, \sigma_{122}, \sigma_{222}) &= (-2, \ 4, \ 2, \ 3.5 \ ).
    \end{align*}

Finally, to recover the off diagonal entries of $\Sigma_1,\Sigma_2$, with sample moments 
        \[ \overline{\mathfrak{m}}_2 =  [\overline m_{11},\overline  m_{21}] = [0.8125, \ 7.75], \]
we   solve the linear system 
\begin{align*}            0.8125 &= 0.25 \cdot (2 + \sigma_{112}) + 0.75 \cdot \sigma_{212} \\
            7.75 &= 0.25 \cdot(-4 - 2 \cdot \sigma_{112}) + 9
        \end{align*}
to find $(\sigma_{112}, \sigma_{212}) = (0.5, \ 0.25)$. We estimate that our density is
    \begin{align*}
        0.25 \cdot \mathcal{N}\Big( \begin{bmatrix} -1 \\ -2 \end{bmatrix}, \begin{bmatrix} 1 & 0.5 \\ 0.5 & 2 \end{bmatrix} \Big) + 0.75 \cdot \mathcal{N} \Big( \begin{bmatrix} 0 \\ 4 \end{bmatrix}, \begin{bmatrix} 3 & 0.25 \\ 0.25 & 3.5 \end{bmatrix} \Big).
    \end{align*}
\end{example}

In \Cref{ex:two-mix}, we use the fact that the moment equations are block triangular. This block triangular structure appears in the higher dimensional setting too. To solve the moment equations, we are really solving $n+1$  subsystems.  
We solve $n$ systems in $2k$ unknowns and one linear system in 
$k\cdot\binom{n}{2}$
unknowns.
This is the reason that our algorithm scales well in the dimension.

Using the homotopy continuation methods discussed in \Cref{sec:finding_all_sols}, we can solve Step~\ref{alg:step3} of \Cref{algorithm:multivariate_density_estimation} 
using homotopy continuation with a binomial start system that is BKK optimal. 
Using these continuation methods for fixed $k$, the number of homotopy paths we need to track is $\mathcal{O}(n)$. 
One observation 
for implementation is that 
Step~\ref{alg:step3} 
can be performed in parallel.

\begin{thm}\label{thm:proof_of_algorithm}
For a generic $\overline \lambda$-weighted Gaussian $k$-mixture model in $\mathbb{R}^n$,
\[ \bar \lambda_1 \mathcal{N}(\mu_1,\Sigma_1) + \cdots + \bar \lambda_k \mathcal{N}(\mu_k, \Sigma_k), \]
  \Cref{algorithm:multivariate_density_estimation} is guaranteed to recover the parameters 
 $\mu_\ell$ and $\Sigma_\ell$, $\ell\in [k]$.
\end{thm}

   Models where all parameters are unknown, including the mixing coefficients, are the most general type of Gaussian mixture model. 
 \Cref{algorithm:multivariate_density_estimation} is easily adaptable for this situation. If the mixing coefficients are not known, then using the block triangular structure highlighted above, one can solve for the mixing coefficients in the first step. 
 The rest of the algorithm then reduces to the case where the mixing coefficients are known. In our experimental results in \Cref{ss: numerical simulation} we take this approach and do not assume that the mixing coefficients are known.
 This is summarized in \Cref{alg:unknown-mix}.

\begin{algorithm}[hbt!]
\DontPrintSemicolon 
\KwIn{
Moments: 
\begin{align*}
\overline{\mathfrak{m}}_{1,i} &:= \left\{\overline{m}_{e_i},\ldots, \overline{m}_{De_i}  \right\} \quad i\in [n]  \quad \text{ and }\quad D=3k+2.
\\
\overline{\mathfrak{m}}_2 &:= \{\overline{m}_{e_i + e_j}, \overline{m}_{2 e_i + e_j},\ldots,\overline{m}_{k e_i + e_j} \ :\  i, j \in [n], \ i \neq j  \}
\end{align*}
from a multivariate Gaussian mixture model:
\[  \lambda_1 \mathcal{N}(\mu_1, \Sigma_1) + \cdots +  \lambda_k \mathcal{N}(\mu_k, \Sigma_k). \]
}
\KwOut{
    Parameters $(\lambda_1,\dots,\lambda_k)\in\Delta_{k-1}$ and
    $\mu_\ell  \in \mathbb{R}^n$, $\Sigma_\ell \succ 0$ for $\ell \in [k]$ such that $\overline{\mathfrak{m}}_{1,i}, \overline{\mathfrak{m}}_2$ are the moments of
    $\sum_{\ell=1}^k  \lambda_\ell \mathcal{N}(\mu_\ell, \Sigma_\ell)$. 
}
{For the
  moments $\overline{\mathfrak{m}}_{1,1}$,
    find all solutions to \eqref{moment_polynomials} with 
    $\lambda\in\Delta_{k-1}$,
    $ \mu_{\ell 1} \in \mathbb{R}$ and $ \sigma_{\ell 1 1 } > 0$ for all $\ell \in [k]$. 
    \;
    Select the solution $(\lambda, \mu, \sigma^2) \in \Delta_{k-1} \times \mathbb{R}^k \times \mathbb{R}_{>0}^k$ consistent with the moments up to $\overline{m}_{De_1}$. \;
}
    {Run \Cref{algorithm:multivariate_density_estimation} with $ \lambda$ as the mixing coefficients.}  
     
     {Return $\lambda$ and $\mu_\ell, \Sigma_\ell$, $\ell\in [k]$.} 
        
\caption{
  Parameter recovery
for mixtures of multivariate Gaussians}\label{alg:unknown-mix}

\end{algorithm}

\begin{rem}
Observe in the proof of \Cref{thm:proof_of_algorithm} the linear system solved in Step~\ref{alg:linear} can be decomposed into $\frac{n(n-1)}{2}$ linear systems of size $k \times k$. Therefore, not only can the for loop in \Cref{algorithm:multivariate_density_estimation} be parallelized, but solving for the off diagonal entries of $\Sigma_1,\ldots,\Sigma_k$ can be as well.
\end{rem}

\begin{rem}
Although \Cref{algorithm:multivariate_density_estimation} uses higher order moments, this order does not exceed what we would need for the univariate case. In addition, it has been shown that there exist algebraic dependencies among lower order moment equations. This causes a failure of algebraic identifiability, thereby complicating the choice of which moment system to consider \cite{amendola2018algebraic}.
\end{rem}

Observe that \Cref{algorithm:multivariate_density_estimation} can be specialized to both the $\lambda$-weighted homoscedastic and $\lambda$-weighted known variance models by applying the results of \Cref{thm:means_unknown_sigma_equal} and \Cref{thm:degree_mus_unknown} to the framework outlined in \Cref{sec:algorithm}. In each case, the number of homotopy paths tracked will be $\frac{(k+1)!}{2}+(n-1)k!$ and $nk!$ respectively. We note again that both situations give algorithms where the number of homotopy paths scales linearly in $n$.

\subsection{Uniform mixtures with equal variances}
In this section, we consider the special case of 
estimating the parameters of a mixture of $k$ Gaussians in $\mathbb{R}^n$ where all of the means $\mu_i \in \mathbb{R}^n$ are unknown, each $\lambda_i = \frac{1}{k}$ and each covariance matrix $\Sigma_i \in \mathbb{R}^{n \times n}$ is equal and known. This is the set up in  \cite{chen2020likelihood,daskalakis2017ten,mei2018landscape, xu2016global}.

Recall from \Cref{cor:identifiability_means_unknown} 
that in one dimension, there is generically a unique solution up to symmetry to \eqref{eq:unknown_mean}.
Therefore, Step~\ref{alg2:step1} of \Cref{algorithm:multivariate_density_estimation_mu_unknown_variance_equal} requires \emph{tracking a single homotopy path}. This is in contrast to Step~\ref{alg:step3} of \Cref{algorithm:multivariate_density_estimation} in which one needs to track $(2k-1)!!k!$ homotopy paths to obtain all complex solutions. 
Further, \Cref{algorithm:multivariate_density_estimation_mu_unknown_variance_equal} requires solving a $k \times k$ linear system $n-1$ times (Step~\ref{alg2:step2}). This is again in contrast to \Cref{algorithm:multivariate_density_estimation} where one needs to solve a nonlinear polynomial system that tracks $(2k-1)!!k!$ paths $n-1$ times (Step~\ref{alg:step3}). 
In both cases, we see that we need to solve $n$ polynomial systems, where $n$ is the dimension of the Gaussian mixture model.

If we consider the tracking of a single homotopy path as unit cost, we consider the number of  paths tracked as the complexity of the algorithm (as is customary in numerical algebraic geometry).
With this choice, \Cref{algorithm:multivariate_density_estimation_mu_unknown_variance_equal} tracks $n$ homotopy paths,
while \Cref{algorithm:multivariate_density_estimation}
tracks $(2k-1)!!k!n$ paths and solves a square linear system of size $\frac{kn(n-1)}{2}$.
This highlights how \Cref{algorithm:multivariate_density_estimation_mu_unknown_variance_equal} 
gives an implementation of the method of moments that is effective for large $n$ and $k$.

\begin{algorithm}[]
\DontPrintSemicolon 
\KwIn{
    The set of sample moments: 
\begin{align*}
\overline{\mathfrak{m}}_1 &:= \{ \overline{m}_{e_1},\ldots,\overline{m}_{ke_1}  \} \\
\overline{\mathfrak{m}}_i &:= \{\overline{m}_{e_i}, \overline{m}_{e_1+e_i}, \overline{m}_{2e_1+e_i},\ldots,\overline{m}_{(k-1)e_1+e_i}  \}, \quad 2 \leq i \leq n
\end{align*}
that are the moments to multivariate Gaussian mixture model:
\[ \frac{1}{k} \mathcal{N}(\mu_1, \overline \Sigma) + \cdots + \frac{1}{k} \mathcal{N}(\mu_k, \overline \Sigma). \]

}
\KwOut{
    Parameters $\mu_\ell  \in \mathbb{R}^n$, such that $\overline{\mathfrak{m}}_i$, $i \in [n]$, are the moments of distribution $\sum_{\ell=1}^k \frac{1}{k} \mathcal{N}(\mu_\ell, \overline \Sigma)$ 
}

Using mixing coefficients $\lambda_\ell = \frac{1}{k}$ for $\ell \in [k]$ and sample moments $\overline{\mathfrak{m}}_1$ solve \eqref{eq:unknown_mean} to obtain $\mu_{\ell 1} \in \mathbb{R}$. \label{alg2:step1}  \;\label{step:alg2-univariate}

Using sample moments $\overline{\mathfrak{m}}_i$
 solve the $k \times k$ linear system in $\mu_{i1},\ldots, \mu_{ik}$ for $2 \leq i \leq n$. \label{alg2:step2}\; 

\caption{Density Estimation for Uniform Mixtures of Multivariate Gaussians with Equal Covariances} \label{algorithm:multivariate_density_estimation_mu_unknown_variance_equal}
\end{algorithm}

\subsection{Numerical simulations}\label{ss: numerical simulation}
As mentioned in \Cref{sec:algorithm}, \Cref{algorithm:multivariate_density_estimation} can be adapted to the case when 
the mixing coefficients $\lambda_\ell$, $\ell \in [k]$ are not known a priori. This is outlined in \Cref{alg:unknown-mix}. The main bottleneck with \Cref{alg:unknown-mix} is the first step which requires one to solve the moment equations for a general Gaussian $k$-mixture model. When $k=2,3,4$ the number of solutions to the corresponding moment system are given in \cite{pearson1894contributions,amendola2016moment,amendola2016solving} respectively. Since the number of complex solutions is known in each of these cases, state of the art polynomial system solvers can exploit the label swapping symmetry to find all solutions quickly, and select the best statistically meaningful one. For $k \geq 5$, the number of solutions to the corresponding moment system is unknown and devising a way solve the Gaussian moment system is open.

\begin{table}[htbp!]
\centering
 \begin{tabular}{|c | c c c c |} 
 \hline
 $n$  & 100 &1{,}000 & 10{,}000 & 100{,}000 \\ [1ex] 
 \hline
 Time (s)  & $0.71$ & $6.17$ & $62.05$ & $650.96$ \\[1ex]
 \hline 
 Error  & $4.1 \times 10^{-13}$ & $5.7 \times 10^{-13}$ & $2.9 \times 10^{-11}$ & $1.8 \times 10^{-9}$ \\ [1ex] 
 \hline 
  Normalized Error  & $1.0 \times 10^{-15}$ & $1.4 \times 10^{-16}$ & $7.3 \times 10^{-16}$ & $4.5 \times 10^{-15}$ \\ [1ex]
 \hline
 \end{tabular}
 \vspace{1mm}
 \caption{Average running time and numerical error running \Cref{alg:unknown-mix} on a mixture of $2$ Gaussians in $\mathbb{R}^n$. The error is $\epsilon = \lVert v - \hat v \rVert_2$ where $v \in \mathbb{R}^{4n+2}$ is a vector of the true parameters and $\hat v$ is a vector of the estimates. The normalized error is $\epsilon \slash (4n+2)$.}\label{table:alg1_results_k_2} 
\end{table}

We perform numerical experiments by running \Cref{alg:unknown-mix} on randomly generated Gaussian mixture models with diagonal covariance matrices. We do not assume the mixing coefficients are known a priori, so we solve for them first and then run \Cref{alg:unknown-mix}.
We use \texttt{HomotopyContinuation.jl} to do all of the polynomial system solving \cite{HomotopyContinuationjulia}. The average running time and error for $k = 2$ are given in \Cref{table:alg1_results_k_2} and for $k = 3$ in \Cref{table:alg1_results_k_3}. 
Overall, we see that the error incurred from doing homotopy continuation is negligible.

\begin{table}[htbp!]
\centering
 \begin{tabular}{|c | c c c c |} 
 \hline
 $n$ &  100 &1{,}000 & 10{,}000 & 100{,}000 \\ [1ex] 
 \hline
 Time (s)  & $10.87$ & $73.74$ & $845.55$ & $8291.84$ \\[1ex]
 \hline 
 Error & $4.6 \times 10^{-12}$ & $1.3 \times 10^{-10}$ & $4.6 \times 10^{-10}$ & $9.6 \times 10^{-9}$ \\ [1ex] 
 \hline 
  Normalized Error & $1.5 \times 10^{-14}$ & $4.2 \times 10^{-14}$ & $1.5 \times 10^{-14}$ & $3.2 \times 10^{-14}$ \\ [1ex]
 \hline
 \end{tabular}
 \vspace{1mm}
 \caption{Average running time and numerical error running \Cref{alg:unknown-mix} on a mixture of $3$ Gaussians in $\mathbb{R}^n$. The error is $\epsilon = \lVert v - \hat v \rVert_2$ where $v \in \mathbb{R}^{6n+3}$ is a vector of the true parameters and $\hat v$ is a vector of the estimates. The normalized error is $\epsilon \slash (6n+3)$.}\label{table:alg1_results_k_3}
\end{table}

Next we test the method of moments using finitely many samples to estimate the moments. We note that in this case the method of moments is not guaranteed to return a statistically meaningful solution, i.e. a real valued solution with positive mixing weights and variances. To test how often we return a statistically meaningful solution and how close the returned solution is to the ground truth we consider two situations. First, we calculate the sample moments using all $N$ samples and perform the standard method of moments. In the second experiment, we leave out $25\%$ of the samples and run the method of moments four times, using the remaining $75\%$ of the samples. We then record the best of the four returned solutions. We evaluate both methods on $1000$ trials of a mixture of two randomly generated Gaussians and record the results in \Cref{tab:mom_failrate} and \Cref{tab:mom_failrate2}.

We generate our random Gaussians by selecting the mixing coefficients, $\lambda$, means, $\mu$, and variances, $\sigma^2$, independently from a $\mathcal{N}(0,1)$ distribution. We then take $\lambda \to \frac{|\lambda|}{\lVert \lambda \rVert_1}$ and $\sigma^2 \to | \sigma^2|$ so that the parameters correspond to a true Gaussian mixture model. We see that as the sample size increases, both versions of the method of moments perform better. In particular we see in \Cref{tab:mom_failrate2} that  instituting a naive sampling regime can improve the efficacy of the method of moments. This suggests the need for future work investigating the influence of different moment estimators. 

\begin{table}[h!]
    \centering
    \begin{tabular}{|c|c|c|c|c|c|c|}
    \hline 
       $N$  & Fail percentage & Max  & Min  & Median  & Mean  & Standard deviation  \\ \hline 
      $1{,}000$   & $28.3\%$ & $113.15$ & $0.03$ & $0.60$ & $1.55$ & $5.13$  \\
      $10{,}000$ & $18.3\%$ & $139.72$ & $0.05$ & $0.35$ & $1.15$ & $5.41$ \\
      $100{,}000$ & $12.3\%$ & $20.86$ & $0.02$ & $0.34$ & $0.83$ & $1.36$ \\
      $1{,}000{,}000$ & $7.6\%$ & $21.93 $ & $0.02 $ & $0.32$ & $0.78$ & $1.28$ \\
      \hline
    \end{tabular}
    \caption{Statistics of the error $\epsilon = \lVert v - \hat{v} \rVert_2$ where $v \in \mathbb{R}^6$ is a vector of the true parameters and $\hat{v} \in \mathbb{R}^6$ is a vector of the estimates from a trial of $1{,}000$ randomly generated Gaussian two mixture models using sample moments calculated with all $N$ samples.}
    \label{tab:mom_failrate}
\end{table}

\begin{table}[h!]
    \centering
    \begin{tabular}{|c|c|c|c|c|c|c|}
    \hline 
       $N$  & Fail percentage & Max  & Min  & Median  & Mean  & Standard deviation  \\ \hline 
      $1{,}000$   & $7.4\%$ & $70.29$ & $0.03$ & $0.49$ & $1.12$ & $2.99$  \\
      $10{,}000$ & $6.8\%$ & $34.86$ & $0.04$ & $0.33$ & $0.83$ & $1.90$  \\
      $100{,}000$ & $5.2\%$ & $9.17$ & $0.02$ & $0.33$ & $0.75$ & $1.06$ \\
      $1{,}000{,}000$ & $3.2\%$ & $14.03 $ & $0.02 $ & $0.31$ & $0.76$ & $1.15$ \\
      \hline
    \end{tabular}
    \caption{Statistics of the error $\epsilon = \lVert v - \hat{v} \rVert_2$ where $v \in \mathbb{R}^6$ is a vector of the true parameters and $\hat{v} \in \mathbb{R}^6$ is a vector of the estimates from a trial of $1{,}000$ randomly generated Gaussian two mixture models using sample moments calculated using $0.75 N$ samples then selecting best solution.}
    \label{tab:mom_failrate2}
\end{table}

\subsection{Comparison against EM}
We conclude by comparing the method of moments against the expectation maximization (EM) algorithm. In this section, we consider the mixture of three Gaussians: $\mathcal{M} = 0.2 \cdot \mathcal{N}(-1,1)+ 0.3 \cdot \mathcal{N}(0, 1.5) + 0.5 \cdot \mathcal{N}(1, 1.25)$. We run $100$ trials where in each trial, we randomly sample $500{,}000$ times from $\mathcal{M}$. 

We consider the standard EM algorithm as outlined in \cite{dempster1977maximum}. The EM algorithm is highly dependent on the relationship between the parameters where one initializes $(\hat{\lambda}, \hat{\mu}, \hat{\sigma}^2)$, versus the true parameters $(\lambda, \mu, \sigma^2)$. We consider three initialization schemes: 
\begin{itemize}
    \item[](EM 1): We initialize $(\hat{\lambda}, \hat{\mu}, \hat{\sigma}^2)$ at the ground truth $(\lambda, \mu, \sigma^2)$.
    \item[] (EM 2): We initialize at a noisy ground truth.
    We choose
    $\hat{\mu_i} = \mu_i + \epsilon_{\mu_i}$,
    $\hat{\sigma}_i^2 = |\sigma_i^2 + \epsilon_{\sigma_i^2}|$,
    and 
    $\hat \lambda\in \Delta_2$ to be proportional to the random vector with coordinates $\lambda_i + \epsilon_{\lambda_i}$
    where $\epsilon_{\lambda_i}, \epsilon_{\mu_i}, \epsilon_{\sigma_i^2}$ are sampled independently from a $\mathcal{N}(0,1)$ distribution.
    \item[] (EM 3): We initialize randomly. We chose 
    $\hat{\mu_i} = \epsilon_{\mu_i}$,
    $\hat{\sigma}_i^2 = |\epsilon_{\sigma_i^2}|$,
    and 
    $\hat \lambda\in \Delta_2$ to be proportional to the random vector with coordinates $\lambda_i + \epsilon_{\lambda_i}$
    where $\epsilon_{\lambda_i}, \epsilon_{\mu_i}, \epsilon_{\sigma_i^2}$ are sampled independently from a $\mathcal{N}(0,1)$ distribution. 
\end{itemize}

We compute five sets of sample moments for the method of moments. One set is created by using all $500{,}000$ samples. The other four sets are created using the sampling regime described in \Cref{ss: numerical simulation}. 
From the five computed solutions, we then select the one closest to the ground truth.
Using this sampling regime, the method of moments returned a statistically meaningful solution in $97$ of the $100$ trials, so we compare EM versus the method of moments on those 97 instances.

To compare the method of moments and EM algorithm, we compare the Euclidean distance between the ground truth parameters and estimated parameters produced by the two estimators. 
The method of moments never returned an estimate closer to the ground truth than EM 1.
This is not surprising as EM 1 initializes the EM algorithm at the ground truth. This  initialization is not realistic in practice. 

Next we note that the method of moments returned a closer estimate to the ground truth than
 EM 2 in $42$ of the trials and EM 3 in $45$ of the trials. 
The statistics for each of the methods in the $97$ successful trials is given in \Cref{tab:mom_v_EM}. 
We observe that the method of moments is competitive with EM 2 and EM 3 and the variance in the quality of the estimator produced by the method of moments is smaller than that of EM.

\begin{table}[h!]
    \centering
    \begin{tabular}{|c|c|c|c|c|c|c|}
    \hline 
       Method   & Max  & Min  & Median  & Mean  & Standard deviation & Time (sec)  \\ \hline 
    MOM   & 1.82 & 0.63 & 0.82 & 0.85 & 0.21 & 1.27  \\
    EM 1    & 0.11 & 0.01 & $0.02$ & $0.03$ & $0.02$ & $10.67$ \\
      EM 2  & 2.49 & 0.05 & $0.82$ & $0.82$ & $0.46$ & $67.63$ \\
      EM 3  & 2.19 & 0.07 & $0.89$ & $0.95$ & $0.50$ & $67.00$ \\
      \hline
    \end{tabular}
    \caption{Statistics of the error $\epsilon = \lVert v - \hat{v} \rVert_2$ where $v \in \mathbb{R}^9$ is a vector of the true parameters and $\hat{v} \in \mathbb{R}^9$ is a vector of the estimates from a trial of $500{,}000$ randomly selected samples of a Gaussian 3 mixture model.}
    \label{tab:mom_v_EM}
\end{table}

In addition, we point out that the running time of the method moments is much better than EM when there are many samples. In particular, using our sampling regime, we had to run the method of moments five times. On average, all five runs of the method of moments still took a fraction of the time it took to run EM, even when initializing EM at the ground~truth. 

\subsection{Outlook and future directions}\label{sec:outlook} 
The goal of this paper was to understand the geometry of the moment equations of Gaussian mixture models.
We characterize when these equations have finitely many solutions (over $\mathbb{C}$) and give an upper bound on this number. As a consequence of our proof techniques, we derive a polyhedral homotopy start system to find all solutions to the moment equations. 
This is important because distinct Gaussian $k$-mixture models can have the same first $3k-1$ moments (see \Cref{ex:eight-moments}).
The natural next question is to determine how many moments are needed 
to get a unique solution to the moment equations. 
We give a new answer to this question in terms of rational identifiability. 
Our main result, \Cref{thm:improv1}, says that generic Gaussian $k$-mixture models are identifiable using $3k+2$ moments, improving 
the Kalai-Moitra-Valiant~\cite{kalai2010efficiently} bound of $4k-2$. 
A future direction is to compare the difference between the number of moments needed for rational indentifiability and the classic identifiability for other statistical models. We expect that tools of tensors and algebraic geometry will continue to provide novel insights to rational identifiability. 

The moment equations are guaranteed to have statistically meaningful solutions if the moments are exact. 
Thus, a natural question to ask is if solving the moment equations is ever practical in the finite sample setting. 
Our computational experiments show initial promise and competitiveness against EM in the finite sample setting. This is only a preview of what is to come, and we think better sampling schemes and new understanding of the moment space will see more impactful moment based methods.

\subsection*{Acknowledgments}
The authors would like to thank Fulvio Gesmundo and Nick Vannieuwenhoven for helpful discussions regarding \Cref{thm:improv1}, 
which took place within the Thematic Research Programme ``Tensors in statistics, optimization and machine learning", University of Warsaw, Excellence Initiative – Research University and the Simons Foundation Award No. 663281 granted to the Institute of Mathematics of the Polish Academy of Sciences for the years 2021-2023.

\newpage

\bibliographystyle{siamplain}
\bibliography{references}    
\pagebreak
\section{Supplementary Material}
\subsection{Proofs from \Cref{sec:preliminaries}} Here we provide the proofs for results given in \Cref{sec:preliminaries}.
\begin{proof}[Proof of \Cref{lem:dependent_polytopes}]
 Fix $w \in W$. By the definition of dependent, we need to show that 
 \[ \dim \Big( \sum_{i=1}^n \init_w(Q_i) \Big) <n. \]
 Since $\init_w(Q_i) \subseteq \mathbb{R}^n$, for all $i \in [n]$, $\sum_{i=1}^n \dim (\init_w(Q_i) ) \leq n$. 
Furthermore, since each $Q_i$ is one dimensional, $\sum_{i=1}^n \dim(\init_w(Q_i)) = n$ if and only if $w$ minimizes all of $Q_i$ for all $i \in [n]$. In other words, $\sum_{i=1}^n \dim(\init_w(Q_i)) = n$ if and only if $\init_w(Q_i) = Q_i$ for all $i \in [n]$.

 Recalling $Q_i = \conv( \{0_n, v_i \})$, one sees $\init_w(Q_i) = Q_i$ if and only if 
 \[ 0 = \langle w, 0_n \rangle = \langle w, v_i \rangle  
\text{ for all } i \in [n].
 \]
  Since $\{v_1,\ldots,v_n \}$ are linearly independent, the only $w$ that satisfies this is $w = 0_n \not\in W$.
\end{proof}
\subsection{Proofs from \Cref{sec:one_d_moment_varieties}}
Here we provide the proofs for all results given in \Cref{sec:one_d_moment_varieties}.
\begin{proof}[Proof of \Cref{lem:induction}]
  We verify both by induction.  For both identities, the base case $k = 1$ is immediate. Suppose $\frac{\partial}{\partial \mu}M_{k-1}( \mu, \sigma^2) = (k-1) \cdot M_{k-2}( \mu, \sigma^2)$. By the recursive relationship \eqref{momentdef1} we have
 \begin{align*}
     \frac{\partial}{\partial \mu} M_k &= M_{k-1} + \mu \frac{\partial}{\partial \mu}M_{k-1} + \sigma^2(k-1) \frac{\partial}{\partial \mu} M_{k-2} \\
     &= M_{k-1} + \mu (k-1) M_{k-2} + \sigma^2(k-1)(k-2) M_{k-3} \\
     &= M_{k-1} + (k-1)M_{k-1} = k M_{k-1}.
 \end{align*}

Similarly, suppose that $\frac{\partial }{\partial \sigma^2} M_{k-1}(\mu,\sigma^2) = \binom{k-1}{2} M_{k-3}(\mu,\sigma^2)$. Using \eqref{momentdef1} again,
 \begin{align*}
     \frac{\partial}{\partial \sigma^2}M_k &= \mu \frac{\partial}{ \partial \sigma^2} M_{k-1} + (k-1)M_{k-2} + \sigma^2(k-1) \frac{\partial}{\partial \sigma^2} M_{k-2} \\
     &= \mu  \binom{k-1}{2}M_{k-3} + (k-1)M_{k-2} + \sigma^2 (k-1)\binom{k-2}{2}M_{k-4} \\
     &= (k-1)M_{k-2} + \binom{k-1}{2} \Big( \mu M_{k-3} + \sigma^2 (k-3)M_{k-4} \Big) \\
     &= (k-1)M_{k-2} + \binom{k-1}{2} M_{k-2} = \binom{k}{2}M_{k-2}. 
 \end{align*} 

\end{proof}

\begin{proof}[Proof of \Cref{lem:algebraic_identifiability_lambdas_known}]
By \cite[Ch.~1, Section 5]{hartshorne_AG} it suffices to show that the Jacobian of \eqref{eq:known_lambda} has a nonzero determinant. 
This Jacobian, denoted by $J_k$, 
is a $2k\times 2k$ matrix with rows indexed by equations and columns indexed by the {unknowns} $\mu_1,\sigma_1,\dots, \mu_k,\sigma_k$:
    \begin{align}
    J_k&=\tilde J_k \cdot \tilde D_k \nonumber \\
    &= \begin{bmatrix}
    \frac{\partial M_{1}}{\partial \mu_1}(\sigma_{1}, \mu_1) &
    \frac{\partial M_{1}}{\partial \sigma_1 }(\sigma_{1}, \mu_1) &
    \dots & 
    \frac{\partial M_{1}}{\partial \mu_k}(\sigma_{k}, \mu_k) &
    \frac{\partial M_{1}}{\partial \sigma_k }(\sigma_{k}, \mu_k) &
\\
    \frac{\partial M_{2}}{\partial \mu_1}(\sigma_{1}, \mu_1) &
    \frac{\partial M_{2}}{\partial \sigma_1 }(\sigma_{1}, \mu_1) &
    \dots & 
    \frac{\partial M_{2}}{\partial \mu_k}(\sigma_{k}, \mu_k) &
    \frac{\partial M_{2}}{\partial \sigma_k }(\sigma_{k}, \mu_k) &
\\
    \vdots & \vdots & \vdots & \vdots & \vdots
\\
    \frac{\partial M_{2k-1}}{\partial \mu_1}(\sigma_{1}, \mu_1) &
    \frac{\partial M_{2k-1}}{\partial \sigma_1 }(\sigma_{1}, \mu_1) &
    \dots & 
    \frac{\partial M_{2k-1}}{\partial \mu_k}(\sigma_{k}, \mu_k) &
    \frac{\partial M_{2k-1}}{\partial \sigma_k }(\sigma_{k}, \mu_k) &
\\
    \frac{\partial M_{2k}}{\partial \mu_1}(\sigma_{1}, \mu_1) &
    \frac{\partial M_{2k}}{\partial \sigma_1 }(\sigma_{1}, \mu_1) &
    \dots & 
    \frac{\partial M_{2k}}{\partial \mu_k}(\sigma_{k}, \mu_k) &
    \frac{\partial M_{2k}}{\partial \sigma_k }(\sigma_{k}, \mu_k) &
    \end{bmatrix} \cdot
    \tilde D_k, \label{eq:Jacobian}
    \end{align}
where $\tilde D_k$ is the diagonal matrix given by {$(\overline \lambda_{1},\overline  \lambda_1,\overline \lambda_2,\overline \lambda_2,\dots,\overline \lambda_k, \overline \lambda_{k})$.}
Note that for nonzero $\lambda_\ell$, $\ell \in [k]$, $J_k$ is full rank if and only if $\tilde J_k$ is full rank. 

We now show that $\tilde J_k$ 
has a determinant that is not identically zero
by induction on $k$. 
When $k = 1$, the determinant of 
 $ \tilde J_1 = 
\left[ \begin{smallmatrix} 1 & 0 \\ 
 2 \mu_1 & 1
 \end{smallmatrix} 
 \right]
 $ 
is  $1$.

Next note that $\mu_k, \sigma_k^2$ only appear in the last two columns of $\tilde J_k$. Further, by \Cref{lem:induction}, the nonzero entries of each row of $\tilde J_k$ have higher degree than the previous row. 
Doing Laplace's cofactor expansion along the last two columns of $\tilde J_k$, we get 
\[\det (\tilde J_k) = \det (\tilde J_{k-1}) \cdot  \mu_k^{2k-2}\sigma_k^{2k-2} + \text{lower order terms}. \] 
By induction, $\det(\tilde J_{k-1})$ is 
not identically zero and neither is  $\det (\tilde J_k)$.
\end{proof}

\begin{proof}[Proof of \Cref{thm:degree_lambdas_unknown}]
Consider the moment equations $f_1^k=0,\ldots,f_{2k}^k=0$
as defined in \eqref{eq:known_lambda}, but now  ordering the unknowns as $(\mu_1,\sigma_1^2,\ldots,\mu_k,\sigma_k^2)$.
Denote $P_i = \newt(f_i^k)$. 

Let $Q_\ell \subset P_\ell$ be the line segment defined as:

\begin{equation*}
    \begin{aligned}
    Q_{2\ell-1} &= \conv (\{0_{2k}, (2 \ell -1) \cdot  e_{2\ell-1} \}), \quad \ell \in [k] \\
    Q_{2\ell} &= \conv (\{0_{2k}, \ell \cdot  e_{2\ell } \}), \quad \ell \in [k],
    \end{aligned}
\end{equation*}
where $0_{2k} \in \mathbb{R}^{2k}$ is the vector of all zeros and $e_\ell \in \mathbb{R}^{2k}$ is the $\ell$th standard basis vector. 
By \Cref{ex:line-segment-det}
we have
$$\mvol(Q_1,\ldots, Q_{2k} ) = (1 \cdot 3 \cdot 5 \cdots (2k-1) )\cdot  (1 \cdot 2 \cdot 3 \cdots k ) =(2k-1)!!k!.$$

We want to use the equivalence of \eqref{item:prop1} and \eqref{item:prop2} in \Cref{prop:guarantee_maximality} to show  
\[\mvol(P_1,\ldots,P_{2k}) = (2k-1)!!k!. \] 
\Cref{bkkbound} then gives that the number of
complex solutions to \eqref{eq:known_lambda} is bounded above by $(2k-1)!!k!$.

For a nonzero vector $w\in \mathbb{R}^{2k}$, 
let
\[ \mathcal{I}_w = \{i \in [2k] : Q_i \cap \init_w(P_i) \, \neq \, \emptyset \}.  \]
We will show for
 each $w \in \mathbb{R}^n \backslash \{0\}$, the collection of polytopes 
    \[\{\init_w (Q_i) : i\in \mathcal{I}_w   \}  \] 
    is dependent.

By \Cref{lem:non_empty_I},
(which we postpone to after this proof)
$\mathcal{I}_w$ is nonempty. 
Since each $Q_i$ is a one dimensional line segment, 
it suffices to show that for any nonzero $w \in \mathbb{R}^{2k}$, $w$ minimizes some $Q_i$ at a unique point for $i \in \mathcal{I}_w$. 
This follows from the definition of dependent since each $Q_i$ is one dimensional, so if $\init_w(Q_i)$ is a single point for some $i \in \mathcal{I}_w$, then $\sum_{i \in \mathcal{I}_w} \init_w(Q_i) < | \mathcal{I}_w|$.

We look at two cases. 
\begin{itemize}[leftmargin=*]
\item First, consider when $2\ell \not\in \mathcal{I}_w$ for all $\ell \in [k]$. 
Since the origin is in $Q_i$, we have $0_{2k} \not\in \init_w(P_{2\ell})$ for all $\ell \in [k]$.
Since $P_{2 \ell} \subset \mathbb{R}_{\geq 0}^{2k}$, this means $\val_w(P_{2\ell})<0$ for all $\ell \in [k]$. 
Hence, $w_i<0$ for some $i\in [2k]$.

Let $i$ be the index of the smallest element of $w$. 
If $i$ is odd, then 
\[
P_i = \conv  \big(\{0_{2k},\frac{i-1}{2} e_2, \ldots, \frac{i-1}{2}e_{2k} , \,\,  i e_1, ,\dots, i e_{2k-1} \} \big) 
\]
is in the non-negative orthant. 
So 
$$
\val_w(P_i)= \min \{0, \frac{i-1}{2}w_2, \dots,  \frac{i-1}{2} w_k, \,\,  i w_1, ,\dots iw_{2k-1}
 \}= i w_i.
$$

So $\init_w(Q_i)=\{ i e_{i} \}$ for $i \in I_w$ and we are done.

Now consider when $i$ is even. Recall,
\[P_{2\ell} =   \conv \big( \{0_{2k}, \ell e_2,\ell e_4,\ldots, \, \ell e_{2k}, 2\ell e_1, 2\ell e_3,\dots, 2\ell e_{2k-1} \}  \big)
= \ell \cdot P_2, \]
and so $\init_w(P_{2 \ell}) = \ell \cdot \init_w(P_2)$ for all $\ell \in [k]$. 
Therefore, $2 \ell \not\in \mathcal{I}_w$ for all $\ell \in [k]$ implies that $P_2$ cannot be minimized at $e_j$ for any even $j$. 
This implies that if $i$ is even, $0>w_{i} > 2 w_j$ for some odd $j$ (otherwise $P_{i}$ 
would be minimized at $\frac{i}{2}e_{i})$. Let $j$ be the index of the smallest odd element of $w$. In this case, $P_j$ would be minimized at $j e_j$ so $j \in \mathcal{I}_w$.
Hence $\init_w(Q_j)$ is $\{j e_j\}$.

\item 
Second, suppose $2 \ell \in \mathcal{I}_w$ for some $\ell \in [k]$. If 
$\init_w(P_{2 \ell})\cap Q_{2\ell}$ is a point, then we are done. 
Otherwise, we may assume $P_{2\ell}$ is minimized by $w$ at
a face containing the line segment
$Q_{2\ell} = \conv(\{0_{2k}, \ell e_{2\ell}\})$.

This means 
\[
0=w^T0_{2k}\leq w^Ta \quad \forall a\in P_{2 \ell}. 
\]
So 
$w_i \geq 0$ for all 
$i\in [2k]$ 
because  the vertices of $P_{2 \ell}$ are scaled standard basis vectors. 
With the fact each $P_i$ is in the non-negative orthant, this implies 
\[
0 = \val_w(P_i)\, \text{ for all } i \in [2k],
\]
so $0_{2k} \in \init_w(P_i)$ for all $i \in [2k]$.

This shows that $\mathcal{I}_w = [2k]$ so by \Cref{lem:dependent_polytopes}, we conclude that the collection of polytopes $\{\init_w(Q_1),\ldots, \init_w(Q_{2k}) \}$ is dependent.
\end{itemize}
\end{proof}

\begin{lem}\label{lem:non_empty_I}
The index set
$\mathcal{I}_w$ as defined in the proof of \Cref{thm:degree_lambdas_unknown} is nonempty for any nonzero $w \in \mathbb{R}^{2k}$.
\end{lem}
\begin{proof}[Proof of \Cref{lem:non_empty_I}]

Recall that 
\[ P_2 = \conv \big( \{0_{2k},e_2,\ldots,  e_{2k},  \, 2e_1,\ldots, 2e_{2k-1} \} \big). \]
We consider three cases, depending on $\init_w(P_2)$.
\begin{itemize}[leftmargin=*]
    \item If either $0_{2k}$ or $e_2$ is in $\init_w(P_2)$, then $Q_2 \cap \init_w(P_2) \,\neq\,  \emptyset$ and hence $2 \in \mathcal{I}_w$.
    
    \item Now suppose $e_j \in \init_w(P_2)$ for some even $j>2$. This means $2w_j \leq w_i$ for all odd $i$ and $w_j \leq w_m$ for all even $m$. Consider
    \[ P_j = \conv \big( \{0_{2k},j e_1, \frac{j}{2} e_2,\ldots, j e_{2k-1}, \frac{j}{2} e_{2k} \} \big). \]
    Then $\frac{j}{2} e_{j} \in \init_w(P_j)$.
    Since $Q_j =\conv(\{0_{2k} ,\frac{j}{2}e_j \})$, 
    we have 
    $j \in \mathcal{I}_w$.
    
    \item 
  On the other hand, now suppose $e_i \in \init_w(P_2)$ for some odd $i \geq 1$. This means $w_i \leq w_j$ for all odd $j$ and $2w_i \leq  w_m$ for all even $m$. Consider 
    \[ P_i = \conv \big( \{ 0_{2k}, \frac{i-1}{2}e_2,\dots, \frac{i-1}{2}e_{2k}, \, i e_1,\ldots,i e_{2k-1} \} \big). \]
    Then $i e_i \in \init_w(P_i)$. Since $Q_i = \conv \big( \{0_{2k}, i e_i \} \big)$, we have $i \in \mathcal{I}_w$.
\end{itemize}
\end{proof}

\begin{proof}[Proof of \Cref{prop:lambda_known_sigma_equal_alg_ident}]
This argument is similar to the one given in \Cref{lem:algebraic_identifiability_lambdas_known}. Again, we consider the Jacobian of \eqref{eq:known_lambda_homoscedastic} with rows indexed by equations and columns indexed by the variables $\sigma^2,\mu_1,\ldots,\mu_k$. It suffices to show that for the Jacobian, $J_k$ is not identically zero. 
We proceed by induction on $k$. When $k = 1$,
 \[ J_1 =  \begin{bmatrix}
 0 & \overline \lambda_1 \\ \overline \lambda_1 & 2 \overline \lambda_1 \mu_1
 \end{bmatrix}  \]
has determinant $- \overline \lambda_1^2$ and is not identically zero because we assume $\overline \lambda_i\neq 0$.
Now consider 
 the $(k+1)\times (k+1)$ matrix 
 $J_k$ for any $k$. By cofactor expansion along the last column,  
 \[\det(J_k) = \mu_k^k \cdot \det( J_{k-1}) +   \text{lower order terms in } \mu_k. \]
 By induction, $\det(J_{k-1})$ is not identically zero. 
 This shows that $\det(J_k)$ is a nonzero univariate polynomial in $\mu_k$ with nonzero leading coefficient $\det(J_{k-1})$. 

\end{proof}

\begin{proof}[Proof of \Cref{thm:means_unknown_sigma_equal}]
For $i \in [k+1]$, let $P_i = \newt(f_i^k) $ where $f_i^k$ is as defined in 
\eqref{eq:known_lambda_homoscedastic}
with variable ordering $(\mu_1,\ldots, \mu_k, \sigma^2)$.

Define $Q_i \subset P_i$ as follows:
\begin{equation}
    \begin{aligned}
    Q_1 &= \conv(\{0_{k+1}, e_1 \} )\\
    Q_2 &= \conv ( \{0_{k+1}, e_{k+1} \} )\\
    Q_{i} &= \conv( \{0_{k+1}, i \cdot e_{i - 1} \}), \qquad 3 \leq i \leq k+1. 
    \end{aligned}
\end{equation}
where $0_{k+1} \in \mathbb{R}^{k+1}$ is the vector of all zeros and $e_i \in \mathbb{R}^{k+1}$ is the $i$th standard basis vector.
By \Cref{ex:line-segment-det}, 
\[ \mvol(Q_1,\ldots,Q_{k+1}) = 1 \cdot 3 \cdot 4 \dots \cdot (k+1) = \frac{(k+1)!}{2}. \]

As in the proof of \Cref{thm:degree_lambdas_unknown}, we want to show the equality $\mvol(P_1,\ldots,P_{k+1}) = \mvol(Q_1,\ldots,Q_{k+1})$ by using the equivalence of \eqref{item:prop1} and \eqref{item:prop2} in \Cref{prop:guarantee_maximality}. 

Let $\mathcal{I}_w$ be the set of indices such that $Q_i$ has a vertex in  $\init_w(P_i)$. Specifically, 
\[ \mathcal{I}_w = \{i \in [k+1] : Q_i \cap \init_w(P_i) \neq \emptyset \}.  \]

By \Cref{lem:non_empty_I_2}, $\mathcal{I}_w$ is nonempty. Now we want to show that for any $w \in \mathbb{Z}^{k+1} \backslash \{0\}$ the nonempty set of polytopes
\[ \{\init_w(Q_i) : i \in \mathcal{I}_w \} \]
is dependent. Since $Q_i$ is a one dimensional line segment, it suffices to show that for any $w$, there exists an $i \in \mathcal{I}_w$ such that $Q_i$ is minimized at a single vertex. We consider two cases.

\begin{itemize}[leftmargin=*]
\item Suppose $2 \in \mathcal{I}_w$. If $\init_w(P_2)$ is a single point, then we are done. Otherwise, assume $P_2$ is minimized at $Q_2 = \conv(\{0_{k+1}, e_{k+1} \})$. This means $w_{k+1} = 0$ and $w_j \geq 0$ for all $j \in [k]$, giving $\val_w(P_i) = 0$ so $0 \in \init_w(P_i)$ for all $i \in [k]$. This shows that $\mathcal{I}_w = [k+1]$. By \Cref{lem:dependent_polytopes}, the collection of polytopes $\{\init_w(Q_1),\ldots, \init_w(Q_{k+1}) \}$ is dependent.
\item Now suppose $2 \not\in \mathcal{I}_w$. If $1 \in \mathcal{I}_w$ and $Q_1$ is minimized at a single vertex then $\{\init_w(Q_i) : i \in \mathcal{I}_w \}$ is dependent, so we are done. If $Q_1$ is not minimized at a single vertex, then $w_1 = 0$ and $w_j \geq 0$ for all $2 \leq j \leq k$. Since $2 \not\in \mathcal{I}_w$, this gives that $w_{k+1} > 2 w_j$ for all $j \in [k]$. Therefore, $\val_w(P_j) = 0$ for all $j \in [k+1]$ which shows that $0 \in \init_w(P_2)$. This contradicts $2 \not\in \mathcal{I}_w$. On the other hand, if $i \in \mathcal{I}_w$ where $i \geq 3$, either $Q_i$ is minimized at a single vertex or $w_{i-1} = 0$ and $w_j \geq 0$ for all $j \in [k+1] \backslash (i-1)$. In the latter case, this shows $\val_w(P_i) = 0$ for all $i \in [k+1]$, contradicting $2 \not\in \mathcal{I}_w$.
\end{itemize}
\end{proof}

\begin{lem}\label{lem:non_empty_I_2}
The index set
$\mathcal{I}_w$ as defined in the proof of \Cref{thm:degree_lambdas_unknown} is nonempty for any nonzero $w \in \mathbb{R}^{2k}$.
\end{lem}
\begin{proof}[Proof of \Cref{lem:non_empty_I_2}] 
If $w$ is in the non-negative orthant, then $w$ minimizes $P_i$ at the origin
for all $i \in [k+1]$. In this case $\mathcal{I}_w = [k+1] \neq \emptyset$.
Now let $i$ be the index of the smallest element of $w$. 
If $i = k+1$ then $w$ minimizes $P_2$ and $Q_2$ at $\{ e_{k+1} \}$. If $i <k+1$ then $w$ minimizes $P_{i+1}$ and $Q_{i+1}$ at $\{(i+1) e_i \}$, which shows $\mathcal{I}_w $ is nonempty. 
\end{proof}

\begin{proof}[Proof of \Cref{lemma:Binomial-BBK-homotopy}]
    By \Cref{bkkbound} the number of solutions for $G(x)=0$ equals the generic number of solutions for a polynomial system with Newton polytopes $\newt(F)$.
    Since $\newt(G)\subseteq\newt(F)$ and $\gamma$ is generic,
    we have
    $\newt(F)=\newt(\gamma(1-t)G+tF)$ for $t\in (0,1]$.
    So the mixed volume, and therefore the number of solutions, of  $(1-t)G+tF$ agrees with the mixed volume of $\newt(F)$.
\end{proof}

\begin{proof}[Proof of \Cref{thm:degree_mus_unknown}]
First we observe that \eqref{eq:unknown_mean} 
for generic moments (as defined by \Cref{eq:generic-moments})
 has finitely many solutions by the same arguments in the proof of \Cref{lem:algebraic_identifiability_lambdas_known} and \Cref{prop:lambda_known_sigma_equal_alg_ident}. This proves the first part of the theorem. 

It follows from the B\'ezout bound that
there are at most $k!$ complex solutions to \eqref{eq:unknown_mean}. We now show that this bound is 
achieved with equality 
for generic moments 
by giving parameter values where there are precisely $k!$ solutions. 

Consider $\overline \lambda_\ell = \frac{1}{k}, \sigma_\ell^2 = 1$ and $\overline{m}_\ell = \sum_{\ell=1}^k \frac{1}{k} M_\ell(1, \ell)$ for $\ell \in [k]$.  It is clear that this has a solution of the form $\boldsymbol{\mu} = (1,2,\ldots,k)$. 
Further, by the same induction argument involving the Jacobian of \eqref{eq:unknown_mean} referenced above, there are finitely many solutions for this set of parameters.
We observe that in this case our solution set has the typical label-swapping symmetry. 
This shows that any action by the symmetric group on $k$ letters, $S_k$, on any solution is also a solution. 
Thus, there are $k!$ solutions to \eqref{eq:unknown_mean} in this case, namely 
{$\{\rho \cdot (1,2,\ldots,k) : \rho \in S_k \}$.}
\end{proof}

\subsection{Proofs from \Cref{sec:highdimstatistics}} Here we give the proofs for results from \Cref{sec:highdimstatistics}.

\begin{proof}[Proof of \Cref{thm:proof_of_algorithm}]

From \Cref{thm:degree_lambdas_unknown}, in each iteration of the for-loop in \Cref{algorithm:multivariate_density_estimation}, we find finitely many prospective choices for each entry of the means $\mu_\ell$ and diagonal entries of $\Sigma_\ell$, $\ell \in [k]$. There is a unique one that agrees with all moments, due to the identifiability result in
\Cref{thm:improv1}.
In other words, after the for-loop
we uniquely recover all mean vectors $\mu_\ell$ and the diagonal entries of $\Sigma_{\ell}$, $\ell \in [k]$. It suffices to show that the linear system in Step~\ref{alg:linear} of \Cref{algorithm:multivariate_density_estimation} has a unique solution. 

Fix $i,j \in [n]$, $i \neq j$, we show that using moment equations $\overline{\mathfrak{m}}_2 = \{ m_{t e_i + e_j} \}$ for $t \in [k]$, the corresponding linear system in the variables $\sigma_{\ell ij}$ for $\ell \in [k]$ generically has a unique solution. 

By \cite[Equation 13]{WILLINK2005271} and the linearity of expectation, using the notation in \eqref{eq:m-to-M} we have
\[
m_{te_i+e_j} = \sum_{\ell = 1}^k
\mu_{\ell j} \cdot M_{te_i}(\mu_\ell,\Sigma_\ell)
+ t \sigma_{\ell ij}\cdot  
M_{(t-1)e_i} (\mu_\ell,\Sigma_\ell) 
\]

After substituting the values we have already obtained into the polynomials $m_{te_i+e_j}$, we get a linear system with a $k \times k$ coefficient matrix that we denote by $A$. 
Indexing the columns of $A$ by $\sigma_{1ij},\ldots,\sigma_{kij}$,
 the $pq$-th entry of $A$ is $p \cdot M_{(p-1)e_i}(\mu_q,\Sigma_q)$.
Observe that the entries of $A$ are polynomials in the parameters $\lambda_\ell, \mu_{\ell i}, \mu_{\ell j}$ for $\ell \in [k]$ and $i, j \in [n]$. Moreover, the entries in row $p$ of $A$ are polynomials of degree $p-1$ and the entries in column $q$ of $A$ contains the subset of parameters corresponding to the $q$th mixture component, $\lambda_q, \mu_{qi}, \mu_{qj}$, $i \in [n]$. Therefore, for generic values of $\mu_\ell, \lambda_\ell, \Sigma_\ell$, these rows will not be linearly dependent 
so the corresponding linear system is full rank, meaning there is a unique solution which recovers the off diagonal entries of $\sigma_{\ell ij}$, $\ell \in [k]$.
\end{proof}

\end{document}